\documentclass{lmcs}

\usepackage[noamsdef,chgbar,linnum,nofnttls,noenmtls]{fmocdmac}

\usepackage{graphicx} 
\usepackage{amsmath}
\usepackage{xspace}
\usepackage{amssymb}
\usepackage{mathrsfs} 
\usepackage{tasks}
\usepackage{tikz}
\usetikzlibrary{automata,positioning,fit,shapes.geometric}
\usepackage{cancel}
\usepackage{comment}

\usepackage{silence}
\let\oldtextsc\textsc
\renewcommand{\textsc}[1]{\textnormal{\oldtextsc{#1}}}

\usetikzlibrary{automata, positioning, arrows.meta}

\newtheorem{remark}{Remark}
\newtheorem{example}[defi]{Example}
\newtheorem{claim}{Claim}
\begin{document}
\nolinenumbers

\title{Automaton-based characterisations of first order logic over infinite trees}
\titlecomment{* This paper is an extended and revised version of \cite{benerecetti2025automaton}}

\author[M.~Benerecetti]{Massimo Benerecetti\rsuper{a}}
\address{\lsuper{a}Università degli Studi di Napoli ``Federico II", Italy}
\email{massimo.benerecetti@unina.it, fabio.mogavero@unina.it}
\author[D.~Della Monica]{Dario Della Monica\rsuper{b}}
\address{\lsuper{b}Università degli Studi di Udine, Italy}
\email{dario.dellamonica@uniud.it, angelo.matteo@uniud.it, gabriele.puppis@uniud.it}
\author[A.~Matteo]{Angelo Matteo\rsuper{b}}
\author[F.~Mogavero]{Fabio Mogavero\rsuper{a}}
\author[G.~Puppis]{Gabriele Puppis\rsuper{b}}


\NewDocumentCommand\Dom{r()}{\mathop{Dom}(#1)}
\NewDocumentCommand\subtree{e{_} m}{#2_{#1}}  

\NewDocumentCommand\AP{}{\mathit{AP}}  
\NewDocumentCommand\Var{}{\mathit{Var}}  

\RenewDocumentCommand\Lang{r()}{\mathscr{L}(#1)}  


\renewcommand{\MSO}{\textsc{MSO}\xspace}
\renewcommand{\MCL}{\textsc{MCL}\xspace}
\renewcommand{\MPL}{\textsc{MPL}\xspace}
\newcommand{\FO}{\textsc{FO}\xspace}
\renewcommand{\WMSO}{\textsc{WMSO}\xspace}
\renewcommand{\WCL}{\textsc{WCL}\xspace}

\newcommand{\ELTL}{\textsc{ELTL}\xspace}
\newcommand{\ECTL}{\textsc{ECTL}\xspace}


\newcommand{\PCTLs}{\textsc{PastCTL$^*$}\xspace}
\newcommand{\CTLs}{\textsc{CTL$^*$}\xspace}
\newcommand{\CTLsf}{\textsc{CTL$^*_f$}\xspace}
\newcommand{\ECTLs}{\textsc{ECTL$^*$}\xspace}
\newcommand{\ECTLsf}{\tzextsc{ECTL$^*_f$}\xspace}
\newcommand{\PCTL}{\textsc{PastCTL}\xspace}
\newcommand{\CTL}{\textsc{CTL}\xspace}
\newcommand{\LTL}{\textsc{LTL}\xspace}
\newcommand{\PLTL}{\textsc{PastLTL}\xspace}
\newcommand{\PolPCTL}{\textsc{PastCTL$_\pm$}\xspace}
\newcommand{\PolCTL}{\textsc{CTL$_\pm$}\xspace}
\newcommand{\PolCTLs}{\textsc{CTL$^*_\pm$}\xspace}
\newcommand{\safeLTL}{\textsc{SafeLTL}\xspace}
\newcommand{\cosafeLTL}{\textsc{coSafeLTL}\xspace}
\newcommand{\safePLTL}{\textsc{SafePastLTL}\xspace}
\newcommand{\cosafePLTL}{\textsc{coSafePastLTL}\xspace}
\newcommand{\safeOLTL}{\textsc{Safe}[\textsc{Past}]\textsc{LTL}\xspace}
\newcommand{\cosafeOLTL}{\textsc{coSafe}[\textsc{Past}]\textsc{LTL}\xspace}
\newcommand{\PDL}{\textsc{PDL}\xspace}

\newcommand{\NBA}{\textsc{NBA}\xspace}
\newcommand{\NBAcf}{\textsc{NBA$_{cf}$}\xspace}
\newcommand{\UBA}{\textsc{UBA}\xspace}
\newcommand{\NCA}{\textsc{NCA}\xspace}
\newcommand{\UCA}{\textsc{UCA}\xspace}
\newcommand{\GTA}{\textsc{GTA}\xspace}
\newcommand{\HTA}{\textsc{HTA}\xspace}
\newcommand{\PHTA}{\textsc{HTA$_\pm$}\xspace}
\newcommand{\WHTA}{\textsc{WHTA}\xspace}

\let\olddiamond\diamond
\let\oldsquare\square
\let\oldbigcirc\bigcirc
\let\diamond\Diamond
\let\square\Box

\newcommand{\uptran}{\ensuremath{\Uparrow}\xspace}
\newcommand{\isroot}{\ensuremath{\mathit{root}}\xspace}
\newcommand{\isnotroot}{\ensuremath{\mathit{non\text{-}root}}\xspace}
\newcommand{\rootvar}{\ensuremath{\rho}\xspace}

\newcommand\sink{{\mathrm{sink}}}
\newcommand\exit{{\mathrm{exit}}}


\renewcommand{\X}{\mathsf{X}\xspace}
\newcommand{\tX}{\Tilde{\mathsf{X}}\xspace}
\renewcommand{\F}{\mathsf{F}\xspace}
\renewcommand{\G}{\mathsf{G}\xspace}
\renewcommand{\U}{\mathsf{U}\xspace}
\newcommand{\W}{\mathsf{W}\xspace}
\renewcommand{\R}{\mathsf{R}\xspace}

\renewcommand{\Y}{\mathsf{Y}\xspace}
\newcommand{\OO}{\mathsf{O}\xspace}
\newcommand{\s}{\mathsf{S}\xspace}
\renewcommand{\B}{\mathsf{B}\xspace}

\renewcommand{\E}{\mathsf{E}\xspace}
\renewcommand{\A}{\mathsf{A}\xspace}
\newcommand{\D}{\mathsf{D}\xspace}
\newcommand{\C}{\mathsf{C}\xspace}

\newcommand{\massimoNote}[1]{{\color{red}MB: #1}}
\newcommand{\darioNote}[1]{{\color{blue}DDM: #1}}
\newcommand{\angeloNote}[1]{{\color{red}AM: #1}}
\newcommand{\fabioNote}[1]{{\color{red}FM: #1}}
\newcommand{\gabrieleNote}[1]{{\color{green!50!black}GP: #1}}

\newcommand{\ranksymbol}{\ensuremath{\mathit{rank}}\xspace}
\newcommand{\rank}[1]{\ensuremath{\ranksymbol(#1)}\xspace}
\newcommand{\rankq}{\rank{q}}
\newcommand{\ranktheta}{\rank{\theta}}
\newcommand{\rankqprime}{\rank{q'}}
\newcommand{\rankthetaprime}{\rank{\theta'}}

\newcommand{\subformula}[1]{\ensuremath{\mathit{sub}(#1)}\xspace}
\newcommand{\subtheta}{\subformula{\theta}}
\newcommand{\subthetaprime}{\subformula{\theta'}}



\begin{abstract}
We study the expressive power of First-Order Logic (\FO) over (unordered)
infinite trees, with the aim of identifying robust characterisations in terms of
branching-time specification formalisms.
While such correspondences are well understood in the linear-time setting, the
branching-time case presents well-known structural challenges.
To this end, we introduce two classes of hesitant tree automata and show that
they capture precisely the expressive power of two branching-time temporal
logics, namely \PolPCTL and \CTLsf, both of which have been previously shown to
be equivalent to \FO over infinite trees.
These results provide uniform automata-theoretic characterisations and yield a
natural normal form for the latter in terms of a new fragment of \CTLs called
\PolCTLs.
As a consequence, we identify a fundamental limitation of \FO in this setting:
along each branch, it can express only properties that are either safety or
co-safety, thereby revealing a sharp expressive boundary for first-order
definability over infinite trees.
\end{abstract}

\maketitle


\section{Introduction}

\emph{Characterisation theorems}~\cite{Flu85} are powerful model-theoretic tools
that offer a principled approach to understanding the intrinsic features of
formal systems.
They allow us to mark the \emph{expressive boundaries} of specification
languages, compare these formalisms \wrt their \emph{descriptive power} on
specific classes of models, and design new languages starting from a given set
of requirements, in the spirit of \emph{Lindstr\"om-style theorems}~\cite{Lin69}
(\eg, based on maximality principles).
They also play a central role in \emph{definability theory}, guiding the
identification of expressive fragments and meaningful extensions of known
logics, thus supporting the selection of suitable languages for the
specification of the correct behaviour of systems in verification and synthesis
tasks.

A foundational distinction exists between \emph{linear-time} and
\emph{branching-time} languages~\cite{MP92,MP95}.
The former capture properties of computations viewed as totally-ordered sets of
events, while the latter account for the branching structure inherent in
concurrent and nondeterministic system behaviours.

The linear-time case, where models are isomorphic to (finite or infinite)
\emph{words}, is by now well understood.
A rich and intertwined network of equivalences connects \emph{predicate logics}
over linear orders with \emph{temporal logics}, such as \LTL~\cite{Pnu77,Pnu81}
and \ELTL~\cite{Wol83}, with \emph{star-free}~\cite{Sch65,PP86} and
\emph{$\omega$-regular}~\cite{Buc62} languages, and with automata-theoretic
models, including \emph{finite}~\cite{Ner58,RS59} and
\emph{B\"uchi}~\cite{Elg61,Buc62,Buc66} automata.
Within this landscape, the \emph{safety} and \emph{co-safety} fragments of
\LTL~\cite{CMP92} occupy a central position: they capture the two most
elementary types of reactive specification, respectively enforcing that
something bad never happens and that something good eventually happens.
All these connections provide deep insights into the structure of definable
properties and lead to optimal decision procedures across different
representations.

By contrast, the branching-time setting remains more fragmented.
Even for \emph{First-Order Logic} (\FO) interpreted over (finite or infinite)
trees many fundamental definability questions remain unsettled.
A longstanding open problem posed by Thomas in the 1980s~\cite{Tho84} asks
whether it is decidable if a given regular tree language is definable in \FO.
This question has been studied under various combinations of tree types
(\emph{ranked/unranked}, \emph{ordered/unordered}) and interpreted vocabularies
(\eg, including only \emph{child}, only \emph{ancestor}, or both relations).
Aside from the positive result for \FO over finite trees with the child
relation~\cite{BS09}, the problem remains open in all other settings.
Efforts to resolve this question have mainly followed \emph{algebraic
approaches}~\cite{Boj21}, inspired by their success in the word case (most
notably the Sch\"utzenberger theorem on star-free languages~\cite{Sch65}).
These approaches rely on the compositionality and structural insights provided
by syntactic algebras.
Despite significant progress, they have provided only partial results, mostly
for classes of finite trees~\cite{Pot95,Boj04,EW10,BSW12} or topologically
simple infinite trees~\cite{BI09, BIS13}.
An alternative and often complementary line of work seeks direct
characterisations of \FO-definable tree languages via automata.
This route, highly successful in the linear-time case, has also led to fruitful
results in the branching-time setting, including a correspondence~\cite{JW96}
between \emph{Monadic Second-Order Logic} (\MSO)~\cite{Rab69}, \emph{Parity Tree
Automata}~\cite{Mos84,EJ91}, and the \emph{Modal \MC}~\cite{Koz83}.
More recently~\cite{BBMP24}, the landscape has expanded to include the
expressive equivalence of \emph{Monadic Chain/Path Logics}
(\MCL/\MPL)~\cite{Tho84,Tho87,HT87}, their temporal \emph{Computation Tree
Logic} counterparts (\ECTLs/\CTLs)~\cite{VW83,EH83,EH85}, and variants of
\emph{Hesitant Tree Automata} (HTA)~\cite{KVW00}.
Within this picture, the boundary between \MSO and \FO has also been probed
through a weaker form of \MCL.
Carreiro and Venema~\cite{CV14} provided a syntactic and automaton-based
characterisation of \emph{Propositional Dynamic Logic} (\PDL)~\cite{FL79}, a
logic that sits strictly below \ECTLs in expressive power, as a fragment of the
modal \MC.
Subsequently, Carreiro~\cite{Car15} showed that \PDL is precisely the
bisimulation-invariant fragment of \emph{Weak Monadic Chain Logic} (\WMCL), a
variant of \MSO that quantifies over finite chains, and introduced a class of
parity automata capturing the expressive power of \WCL over trees.

In this work, we continue along this line of development, by providing the
first complete automaton-based
characterisation, to the best of our knowledge, of first-order logic with the descendant relation over unranked
unordered infinite trees.
Throughout this work, all temporal logics are implicitly considered in their
\emph{counting extensions}~\cite{Fin72}, which allow formulas to express
properties of the number of successors satisfying a given condition.
Our approach builds on previous results for two branching-time temporal logics,
namely a \emph{fragment of Computation Tree Logic with past}, denoted
\emph{\PolPCTL}, and the \emph{Full Computation Tree Logic with finite path
quantification}, denoted \emph{\CTLsf}.
In~\cite{Sch92a,BBMP23} these logics were shown to be expressively equivalent to
\FO when interpreted on unranked unordered infinite trees.
For each of these two logics, we introduce corresponding variants of hesitant
graded automata, called \emph{Two-Way Linear Polarised Hesitant Tree Automata}
and \emph{Counter-free Visible Polarised Hesitant Tree Automata}, and prove
that they capture precisely the expressive power of the considered logics, and
therefore of \FO as well.
This establishes a full mutual equivalence between logics and automata.
These characterisations also uncover a \emph{polarised property} for both
temporal logics, revealing a noteworthy semantic feature of \FO over infinite
trees: formulas that quantify existentially on branches can only express
\emph{co-safety} properties, while those quantifying universally correspond
exclusively to \emph{safety} properties.
This polarisation property is not merely a syntactic artefact: it arises
independently in \PolPCTL, where it can be enforced purely by syntactic means,
and in \CTLsf, where a dedicated lemma
(analogous to a path-quantifier normalisation) is needed to show that this logic is equivalent to a syntactic fragment of \CTLs, denoted \PolCTLs.
This observation aligns with the already discussed results by Carreiro and
Venema.
As a further contribution, we provide a direct automaton-based characterisation
of the safety and co-safety fragments of \LTL themselves.
Specifically, we show that counter-free looping \emph{universal} B\"uchi automata are
expressively equivalent to \safeLTL, and that counter-free looping \emph{nondeterministic}
co-B\"uchi automata are expressively equivalent to \cosafeLTL.
While the connection between counter-freeness and aperiodicity, and the
general role of B\"uchi and co-B\"uchi conditions in the word setting, are
well known, this precise automaton-based characterisation of the safety and
co-safety fragments does not appear to have been stated explicitly in the
literature, and serves as a self-contained result of independent interest.
It also provides a clean linear-time analogue to the polarised hesitant
automata classes we introduce for trees. 

The main results of the paper are described in Figure \ref{fig:expressive-equivalences}.

\begin{figure}[t]
  \centering
\begin{tikzpicture}[
  formalism/.style={
    draw,
    rounded corners=2pt,
    fill=black!3,
    align=center,
    minimum width=2.9cm,
    minimum height=1cm,
    inner sep=5pt,
    font=\small
  },
  edge/.style={line width=.5pt},
  relation/.style={fill=white, inner sep=2pt, font=\large},
  reference/.style={
    fill=white, inner sep=1pt, font=\scriptsize, align=center
  }
]
  \node[formalism] (ppctl)   at (-4.8, 0) {\(\PolPCTL\)};
  \node[formalism] (fo)      at ( 0,   0) {\(\FO\)};
  \node[formalism] (ctlsf)   at ( 4.8, 0) {\(\CTLsf\)};
  \node[formalism] (twoway)  at (-4.8,-3) {Two-way linear\\\(\PHTA\)};
  \node[formalism] (visible) at ( 0,  -3) {Counter-free visible\\\(\PHTA\)};
  \node[formalism] (polctls) at ( 4.8,-3) {\(\PolCTLs\)};

  \draw[edge] (fo) -- (ctlsf)
    node[midway,relation] {\(\equiv\)}
    node[midway,reference,above=4pt] {\textit{\cite{BBMP23}}};

  \draw[edge] (ppctl) -- (fo)
    node[midway,relation] {\(\equiv\)}
    node[midway,reference,above=4pt] {\textit{\cite{Sch92a}}};

  \draw[edge] (ppctl) -- (twoway)
    node[midway,relation] {\(\equiv\)}
    node[midway,reference,right=5pt] {\textit{Thm. \ref{cor:twowaypctl}}};

  \draw[edge] (visible) -- (polctls)
    node[midway,relation] {\(\equiv\)}
    node[midway,reference,below=5pt] {\textit{Thm. \ref{cor:visiblectl}}};

  \draw[edge] (ctlsf) -- (polctls)
    node[midway,relation] {\(\equiv\)}
    node[midway,reference,left=5pt] {\textit{Thm \ref{normal form CTLsf}}};

  \draw[edge] (twoway) -- (visible)
    node[midway,relation] {\(\equiv\)}
    node[midway,reference,below=5pt] {\textit{Cor. \ref{automata equivalence}}};
\end{tikzpicture}
  \caption{Main results of the paper and some known results.}
  \label{fig:expressive-equivalences}
\end{figure}

\paragraph{Related work.}

In earlier work, Boja{\'n}czyk~\cite{Boj04} showed that, over finite binary
trees, \FO with child and ancestor relations is equivalent to a \emph{cascade
product} of so-called \emph{aperiodic wordsum automata}.
While related in spirit, this result targets a different logic and a different
class of structures.
More recently, Ford~\cite{For19} focused on the same tree structures that are
considered here, and introduced the class of \emph{antisymmetric path parity
automata}, which are shown to be no more expressive than \FO.
However, that work does not provide a translation from \FO to automata, leaving
the equivalence question open.
On the side of temporal logics, the syntax of \PolPCTL is related to, but
distinct from, Schlingloff's original presentation of \CTL with
past~\cite{Sch92a}; for this reason, we provide a careful comparison of these
two.

\paragraph{Organisation.}

This paper is a considerably revised and extended version of~\cite{benerecetti2025automaton}.
In particular, missing proofs are added, Section \ref{word automata char} is completely new,
Section \ref{second automaton} has a new example to clarify
the intuition behind the \emph{visibility} condition, and Section \ref{discussion} has been expanded with a long discussion on open problems. We detail the organisation section by section.

In Section~2 we give the necessary preliminaries on words and trees, first-order
logic, tree temporal logics, linear temporal logics including \LTL and \PLTL and
their safety and co-safety fragments, and classical word automata.
In Section~3 we introduce the branching-time temporal logics equivalent to
\FO: we define \PolPCTL and establish its expressiveness results, introduce
\CTLsf, and derive its normal form \PolCTLs.
Section~4 is devoted to the definition of graded tree automata and their
hesitant and polarised restrictions.
In Section~5 we prove the automaton-based characterisation of \PolPCTL via
two-way linear polarised hesitant tree automata, and establish their equivalence
with \FO.
Section~6 provides an automaton-based characterisation of the safety and
co-safety fragments of \LTL via counter-free looping universal B\"uchi automata
and nondeterministic co-B\"uchi automata.
In Section~7 we prove the characterisation of \CTLsf via counter-free visible
polarised hesitant tree automata, completing the second correspondence with \FO.
Finally, Section~8 discusses the results, outlines open problems, and identifies
directions for future work.



\section{Preliminaries}\label{preliminaries}

\paragraph{Words and trees.}
A \emph{finite word} over an alphabet $\Sigma$ is a finite sequence $w=a_1\dots a_n$ of letters from $\Sigma$. We write $w(i)$ to denote the letter from $\Sigma$ in the $i$-th position of $w$.
We denote by $\Sigma^*$ the set of all finite words over $\Sigma$, which includes the empty word $\varepsilon$.
A \emph{word language} is any subset of $\Sigma^*$.
Similarly, we define an \emph{$\omega$-word} over $\Sigma$ as an infinite sequence $w=a_1 a_2\dots$ of letters from $\Sigma$;
we denote by $\Sigma^\omega$ the set of all $\omega$-words over $\Sigma$, and call \emph{$\omega$-language}
any subset of $\Sigma^\omega$.

We shall also work with trees, and precisely with unranked, unordered, leafless trees.
By \emph{unranked} we mean that every node can have any finite number of immediate successors;
by \emph{unordered} we mean that trees cannot be distinguished by just reordering the successors of a node;
by \emph{leafless} we mean that every node has at least one successor.
We adopt the usual terminology for trees: ancestor, descendant, parent, sibling, root, etc.
We denote by $\Dom(t)$ the set of nodes of a tree $t$.
Given an alphabet $\Sigma$, a \emph{$\Sigma$-labelled tree} ($\Sigma$-tree for
short) is a tree $t$ equipped with a labelling $\lambda_t: \Dom(t) \to \Sigma$,
which associates a letter from $\Sigma$ with each node of $t$.
By a slight abuse of notation, we also denote the resulting $\Sigma$-tree by $t$,
and given $v\in\Dom(t)$, we write $t(v)$ for the label assigned to $v$ by $\lambda_t$.
A \emph{tree language} over $\Sigma$ is a set of $\Sigma$-trees.
We denote by $T_\Sigma$ the language of all $\Sigma$-trees.

A \emph{path} of a tree (or a $\Sigma$-tree) $t$ is a finite or infinite sequence $\pi = v_0 v_1 \dots$
of nodes totally ordered by the successor relation (i.e., each $v_i$ is a successor of $v_{i-1}$ for all $i>0$),
but not necessarily starting from the root. If the path starts from the root it is called \emph{initial}.
We denote by $\pi(i)$ the $i$-th node $v_i$ along the path $\pi = v_0 v_1 \dots$.
A \emph{maximal path} is a path that cannot be prolonged to the right;
in particular, because we restrict ourselves to leafless trees, a maximal path is necessarily infinite.
%
%
We denote by $\subtree_u{t}$ the \emph{subtree} of $t$ rooted at node $u$,
namely, the tree $t$ restricted to the set of all descendants of $u$, including $u$ itself.
Of course, if $t$ is a $\Sigma$-tree, then the labelling of the nodes of
$\subtree_u{t}$
is inherited from $t$, i.e., $\subtree_u{t}(v) = t(v)$ for all $v\in\Dom(\subtree_u{t})$.

\paragraph{First-order logic. }
We introduce \emph{monadic first-order logic on trees with the ancestor relation} (\FO for short).
Let us fix a finite set $\AP$ of atomic propositions and a countable set $\Var$ of variables.
Formulae of \FO are generated by the following grammar,
where $p$ ranges over $\AP$ and $x, x'$ range over $\Var$:
\[
\varphi ~::=~
  (x = x')
  ~\bigm|~
  (x \le x')
  ~\bigm|~
  p(x)
  ~\bigm|~
  \neg \varphi
  ~\bigm|~
  \varphi \lor \varphi
  ~\bigm|~
  \exists x. ~ \varphi
\]
A variable is \emph{free} in a formula if it does not occur under the scope of any quantifier.
A \emph{sentence} is a formula without free variables.

\FO formulae are interpreted over tree structures of the form
$\mathcal{T} = (t,\zeta)$, where $t$ is a $\Sigma$-tree, with $\Sigma=2^{\AP}$,
and $\zeta: \Var \to \Dom(t)$ associates each variable with a node of $t$.
This results in a satisfaction relation $\models$ between \FO formulae and structures
that is defined inductively as follows:
\begin{itemize}
\item $\mathcal{T} \models (x=x')$ iff $\zeta$ maps $x$ and $x'$ to the same node,
\item $\mathcal{T} \models (x \le x')$ iff $\zeta(x)$ is an ancestor of $\zeta(x')$,
\item $\mathcal{T} \models p(x)$ iff $p\in t(\zeta(x))$,
\item $\mathcal{T} \models \neg\varphi$ iff $\mathcal{T} \not\models \varphi$,
\item $\mathcal{T} \models \varphi_1 \lor \varphi_2$ iff $\mathcal{T} \models \varphi_1$ or $\mathcal{T} \models \varphi_2$,
\item $\mathcal{T} \models \exists x. ~ \varphi$ iff there is a node $v\in\Dom(t)$ such that
      $\mathcal{T}' \models \varphi$, where $\mathcal{T}'=(t,\zeta')$ and $\zeta'$ is defined
      by $\zeta'(x)=v$ and $\zeta'(x')=\zeta(x')$ for all $x'\in\Var\setminus\{x\}$.
\end{itemize}
Note that sentences can be interpreted directly over trees, as the mapping $\zeta$ for the
variables is immaterial.

\paragraph{Tree Temporal logics}
We will consider several temporal logics that can be interpreted over tree
structures.
To enable correspondences with \FO formulae, \emph{graded} variants of temporal logics
are often used (see, for instance, \cite{moller2003counting}).
These variants extend the standard formalisms with \emph{counting modalities},
which express lower and upper bounds on the number of successors of a node
that satisfy a given property (e.g., ``at least two distinct successors of $x$ satisfy $\varphi$'').
In the literature, graded extensions are often marked explicitly in the names of the temporal logics;
however, to avoid notational clutter, we shall not adopt this convention here, and rather assume
that all temporal logics are graded.

The first and most expressive temporal logic that we consider here is
\emph{Full Computation Tree Logic with Past}, denoted \PCTLs and defined just below.

\begin{defi}[\PCTLs syntax]
\PCTLs formulae are generated according to the following grammar:
\begin{align*}
\varphi ~::= &~
  p
  ~\bigm|~
  \neg \varphi
  ~\bigm|~
  \varphi \lor \varphi
  ~\bigm|~
  {\D^n\varphi}
  ~\bigm|~
  \E \psi
\\
\psi ~::= &~
  \varphi
  ~\bigm|~
  \neg \psi
  ~\bigm|~
  \psi \lor \psi
  ~\bigm|~
  \X \psi
  ~\bigm|~
  \Y \psi
  ~\bigm|~
  \psi \U \psi
  ~\bigm|~
  \psi \R \psi
  ~\bigm|~
  \psi \s \psi
\end{align*}
where, as usual, $p$ ranges over a set $\AP$ of atomic propositions.
The formulae generated from $\varphi$ are called \emph{state formulae},
while those generated from $\psi$ are called \emph{path formulae}.
The \emph{counting modality} \(\D^n\varphi\) expresses that at least $n \ge1 $ distinct successors of the current node satisfy $\varphi$. The \emph{existential path quantifier} $\E\psi$ expresses that there exists a branch starting at the current node along which the path formula $\psi$ holds. The future temporal operators $\X$, $\U$, and $\R$ stand for \emph{next}, \emph{until}, and \emph{release}, respectively. The formula $\X\psi$ requires $\psi$ to hold at the next node along the selected branch. The formula $\psi_1\U\psi_2$ requires $\psi_2$ to hold at some current or future position, with $\psi_1$ holding at every position from the current one up to, but excluding, that position. The formula $\psi_1\R\psi_2$ requires $\psi_2$ to hold until and including a position where $\psi_1$ holds; if $\psi_1$ never holds, then $\psi_2$ must hold forever. The past temporal operators $\Y$ and $\s$ stand for \emph{yesterday} and \emph{since}, respectively. The formula $\Y\psi$ requires $\psi$ to hold at the immediately preceding node. The formula $\psi_1\s\psi_2$ requires $\psi_2$ to hold at some current or past position, with $\psi_1$ holding at every subsequent position up to and including the current one.
\end{defi}

\noindent
We adopt the following shorthands:
$\top = p \vee \neg p$ for any $p\in\AP$,
$\F\psi = \top \U \psi$,
$\G\psi = \bot \R \psi$,
$\psi_1 \W \psi_2 = \psi_2 \R (\psi_1 \lor \psi_2)$, $\OO\psi = \top \s \psi$,
$\A\psi = \neg \E\neg\psi$.
Note that the release operator $\R$ is redundant in the above grammar, since it is
definable as the dual of $\U$ ($\psi_1 \R \psi_2 = \neg(\neg\psi_1 \U
\neg\psi_2)$).
We still include it as a primitive operator since it is not definable in some
of the fragments that we are going to consider next, in particular in \PCTL and \CTL.

As usual for temporal logics, state formulae of $\PCTLs$ are interpreted over \emph{pointed trees},
i.e., pairs $(t,u)$ consisting of a $\Sigma$-tree and distinguished node $u$ in it, where $\Sigma=2^{\AP}$.
Similarly, path formulae are interpreted over \emph{path-pointed trees}, i.e., triples $(t,\pi,i)$
consisting of a $\Sigma$-tree, a distinguished 
maximal path $\pi$ in it, and a distinguished position
$i$ on that path (recall that, in this case, $i$ ranges over the natural numbers).

\begin{defi}[\PCTLs semantics]
The satisfaction relations $\models_s$ and $\models_p$, respectively between pointed trees and
state formulae and between path-pointed trees and path formulae, are defined by a mutual induction as follows:
\begin{itemize}
\item $(t,u) \models_s p$ iff $p \in t(u)$,
\item $(t,u) \models_s \neg\varphi$ iff $(t,u) \not\models_s \varphi$,
\item $(t,u) \models_s \varphi_1 \lor \varphi_2$ iff $(t,u) \models_s \varphi_1$ or $(t,u) \models_s \varphi_2$,
\item $(t,u) \models_s \D^n \varphi$ iff $u$ has at least $n$ distinct successors $v_1, \dots, v_n$ such that $(t, v_i) \models_s \varphi$ for all $i=1,\dots,n$,
\item $(t,u) \models_s \E \psi$ iff there is an infinite path $\pi$ starting at the root%
             \footnote{Here we adopt a history-preserving semantics for path quantifiers,
                       which is rather standard in the literature of tree temporal logics with past modalities.
                       On the other hand, for tree temporal logics without past modalities, the literature is less uniform,
                       and many works adopt a forgetful semantics of $\E\psi$, which erases the
                       information about the access path to the current node $u$ and evaluates
                       the subformula $\psi$ over some infinite path starting at $u$.
                       The two semantics are clearly distinguished by formulae with past modalities.
                       Fortunately, this discrepancy is harmless for our purposes:
                       all the results from the literature we rely on that adopt a forgetful semantics
                       concern future-only logics, for which the distinction is immaterial.}
             of $t$ and visiting $u$ at position $i$ such that $(t,\pi,i) \models_p \psi$,
\end{itemize}
\begin{itemize}
\item $(t,\pi,i) \models_p \varphi$ iff $(t,\pi(i)) \models_s \varphi$,
\item $(t,\pi,i) \models_p \neg \psi$ iff $(t,\pi,i) \not\models_p \psi$,
\item $(t,\pi,i) \models_p \psi_1 \lor \psi_2$ iff $(t,\pi,i) \models_p \psi_1$ or $(t,\pi,i) \models_p \psi_2$,
\item $(t,\pi,i) \models_p \X \psi$ iff $(t,\pi,i+1) \models_p \psi$,
\item $(t,\pi,i) \models_p \Y \psi$ iff $i > 0$ and $(t,\pi,i-1) \models_p \psi$,
\item $(t,\pi,i) \models_p \psi_1 \U \psi_2$ iff there is $j \ge i$ such that $(t,\pi,j) \models_p \psi_2$ and for all $k=i,\dots,j-1$, $(t,\pi,k) \models_p \psi_1$,
\item $(t,\pi,i) \models_p \psi_1 \R \psi_2$ iff $(t,\pi,j) \models_p \psi_2$ holds for all $j \geq i$, or there is there is $j \ge i$ such that $(t,\pi,j) \models_p \psi_1 \land \psi_2$ and for all $k=i,\dots,j-1$, $(t,\pi,k) \models_p \psi_2$,
\item $(t,\pi,i) \models_p \psi_1 \s \psi_2$ iff there is $j \le i$ such that $(t,\pi,j) \models_p \psi_2$ and for all $k=j+1,\dots,i$, $(t,\pi,k) \models_p \psi_1$.
\end{itemize}
\end{defi}

\noindent
For simplicity, we shall use $\models$ for either $\models_s$ or $\models_p$, depending
on which type of structure and formula we are considering.
Two state (resp.~path) formulae $\psi_1,\psi_2$ are \emph{full state equivalent} if,
for every pointed tree $(t,u)$, we have $(t,u)\models_s \psi_1$ iff $(t,u) \models_s \psi_2$
(resp.~for every path-pointed tree $(t,\pi,i)$, we have $(t,\pi,i)\models_p \psi_1$ iff $(t,\pi,i) \models_p \psi_2$).
By a slight abuse of notation, we sometimes write $t \models \varphi$ to mean $(t,\varepsilon_t) \models \varphi$
(recall that $\varepsilon_t$ is the root of $t$), and in this case we also say that
$t$ is a \emph{model} of $\varphi$.
We denote by $\Lang(\varphi)$ the set of all models $t$ of $\varphi$, and accordingly
say that two state formulae $\varphi$ and $\varphi'$ are \emph{initially equivalent}
if $\Lang(\varphi) = \Lang(\varphi')$.
Note that in this case we are using a notion of equivalence, called \emph{initial equivalence},
that evaluates state formulae at the root of an arbitrary tree.
In general, notably in the presence of past modal operators,
this equivalence is coarser than full state equivalence.
On the one hand, we will often use full state equivalence for comparing formulae
and for performing syntactic translations from one formalism to another -- this
is because full state equivalence is naturally compatible with substitutions of
subformulae.
On the other hand, when comparing the expressive power of logical formalisms, we
will mostly use initial equivalence between state formulae.
For example, we say that one formalism is \emph{as expressive as} another formalism
if there are effective translations of each formula from each formalism
to an initially equivalent formula of the other formalism.
Similarly, we say that one formalism is \emph{less expressive than} another formalism
if there is an effective translation from each formula of the former to an
initially equivalent formula of the latter, but not the other way around.

%


\medskip
The following are well known fragments of \PCTLs:
\begin{itemize}
    \item \CTLs is the variant of \PCTLs without the past operators $\Y$ and $\s$;
    \item \PCTL is the variant of \PCTLs in which every future temporal operator (i.e., $\X$, $\U$, and $\R$) must be immediately preceded by the path quantifier $\E$;
    \item \CTL is the intersection of the previous two fragments, forbidding
past operators and enforcing a strict alternation between future temporal
operators and path quantifiers.
\end{itemize}

\begin{example}
$\E(\G p)$ is a \CTL formula;
$\E(\F(p\s q))$ is a \PCTL formula, but neither a \CTL nor a \CTLs formula;
$\E(\G\F p)$ is a \CTLs formula, but not a \PCTL formula;
finally, $\A (\F\G(p \s q))$ is a proper \PCTLs formula that does not belong to any of the previous fragments.
\end{example}

Of course, a syntactic fragment does not need to be strictly less expressive than the full logical formalism.
Below, we recall the relative expressiveness of the logics \PCTLs, \CTLs, \PCTL, \CTL.

\begin{prop} \cite[Theorem 3.5]{HT87} \cite[Section 4]{laroussinie1995hierarchy}
\label{PCTLs-expressiveness}
\begin{itemize}
    \item \CTLs is as expressive as \PCTLs;
    \item \PCTL is strictly less expressive than \CTLs;
    \item \CTL is strictly less expressive than \PCTL.
\end{itemize}
\end{prop}

\noindent
Basically, adding past operators does not increase expressiveness for \CTLs, but
it does so for the weaker logic \CTL.

\paragraph{Linear Temporal logics}
By seeing $\omega$-words over $\Sigma$ as degenerate cases of $\Sigma$-trees, where every node has exactly
one successor, one can interpret \FO sentences and $\PCTLs$ formulae directly over $\omega$-words.
However, because in an $\omega$-word there is only one maximal path, the quantifiers $\D^n$ and $\E$ of $\PCTLs$
are of no use, and similarly for the Boolean connectives appearing in the grammar of a state formula.
In particular, one can define a fragment of $\PCTLs$ where the only state formulae allowed are the atomic propositions,
giving rise to a simpler grammar for the so-called \emph{Linear Temporal Logic with Past} (\PLTL for short):
\[
\psi ~::=~
  p
  ~\bigm|~
  \neg \psi
  ~\bigm|~
  \psi \lor \psi
  ~\bigm|~
  \X \psi
  ~\bigm|~
  \Y \psi
  ~\bigm|~
  \psi \U \psi
  ~\bigm|~
  \psi \R \psi
  ~\bigm|~
  \psi \s \psi.
\]
The \emph{pure future} fragment of $\PLTL$, denoted $\LTL$, is obtained by
further disallowing the past operators $\Y$ and $\s$.
Analogously, the \emph{pure past} fragment of $\PLTL$ is obtained by disallowing
the future operators $\X$, $\U$, and $\R$.
The semantics of both $\PLTL$ and $\LTL$ is inherited from $\PCTLs$, and the definitions of
full state equivalence, initial equivalence, and expressiveness are as before,
with the only difference that in this setting the models are restricted to be $\omega$-words.
Under these assumptions and thanks to Proposition \ref{PCTLs-expressiveness}, we easily
derive that $\LTL$ and $\PLTL$ have the same expressive power. In fact, it is also known
that, over models that are $\omega$-words, $\LTL$ formulae are as expressive as \FO sentences:

\begin{prop} \cite{kamp1968tense} \cite[Theorem 2.1]{gabbay1980temporal} \label{PLTL-expressiveness}
Under initial equivalence and when the models are restricted to be $\omega$-words,
the logical formalisms \LTL, \PLTL, and \FO have the same expressive power.
\end{prop}


We will study even more restricted fragments of \PLTL and \LTL, that capture
precisely the \emph{safety} and \emph{co-safety} properties
\cite{alpern1987recognizing}.
These fragments are abbreviated \safeOLTL and \cosafeOLTL, where the optional
``\textsc{Past}'' component of the name denotes the availability of past
operators, and are defined as follows.
%
%
The fragment \safePLTL (resp., \cosafePLTL) contains all \PLTL formulae of the form $\G \psi$ (resp., $\F \psi$), where $\psi$ is a pure past \PLTL formula.
The fragment \safeLTL (resp., \cosafeLTL) contains \LTL formulae in negation normal form ---namely, negation appears only in front of atomic propositions---
that use only the operators $\X$ and $\R$
%
(resp., $\X$ and $\U$).

\begin{thmC}[{\cite[Theorems 1 and 8]{CMP92}}] \label{safeLTL equivalent to safePLTL}
Under initial equivalence, \safePLTL is as expressive as \safeLTL, and
\cosafePLTL is as expressive as \cosafeLTL.
\end{thmC}

\paragraph{Word automata.}
A \emph{nondeterministic B\"uchi word automaton} (\NBA) is a tuple
$\mathcal{A} = (Q, \Sigma, \Delta, q_I, F)$,
where $Q$ is a finite set of states,
$\Sigma$ is a finite alphabet,
$\Delta \subseteq Q\times\Sigma\times Q$
is a transition relation,
$q_I\in Q$ is an initial state, and
$F \subseteq Q$ is a set of \emph{final states}.
An NBA is \emph{deterministic} if $\Delta$ can be seen
as a partial function from $Q\times\Sigma$ to $Q$.

A \emph{path} of $\mathcal{A}$ on a finite or infinite word $w = \sigma_1 \sigma_2 \dots$
is a sequence of states $\rho = q_0 q_1 \dots$ of length $|w|+1$ (if $w$ is infinite, then the path is also infinite)
such that $(q_{i-1},\sigma_{i},q_i)\in\Delta$ for all positions $i$ in $w$.
Such a path $\rho$ is called a \emph{run} of $\mathcal{A}$ on $w$ if it starts with the initial state $q_I$ of the automaton.
We denote by $\inf(\rho)$ the set of states visited infinitely often in an infinite run $\rho$,
and say that $\rho$ is \emph{successful} if there is a $q \in F$ such that $q \in \inf(\rho)$.
An infinite word $w$ is \emph{accepted} by $\mathcal{A}$ if there is a successful run
of $\mathcal{A}$ on $w$.
The language $\Lang(\mathcal{A})$ \emph{recognised} by $\mathcal{A}$ is
the set of words accepted by $\mathcal{A}$.
%
%
Two automata are \emph{equivalent} if they recognise the same language.

Variants of NBAs are obtained by modifying the acceptance condition.
One way is to simply turn existential nondeterminism into universal nondeterminism,
that is, to declare that the automaton accepts an infinite word $w$ whenever \emph{every} run
on $w$ is successful; we call these models \emph{universal automata} (e.g., we have
\emph{universal B\"uchi word automata}, abbreviated \UBA).
Another way is to replace a B\"uchi acceptance condition with a coB\"uchi one:
formally, a run $\rho$ of a \emph{nondeterministic coB\"uchi automaton}
(\emph{NCA} for short) is successful if there is no $q\in F$ such that
$q\in\inf(\rho)$; accordingly, the states of $F$ are called \emph{rejecting
states}.

An automaton $\mathcal{A}$ is \emph{counter-free} if for all states $q$,
all finite words $w$, and all natural numbers $n$,
if $\mathcal{A}$ admits a path on $w^n$ that starts and ends at
$q$, then it also admits a path on $w$ that starts and ends at $q$.

\begin{prop} \cite[Theorem 1.1]{diekert2008first}
The languages recognisable by nondeterministic, counter-free B\"uchi automata
are exactly those definable in \FO.
\end{prop}




\section{Temporal logics equivalent to \FO} \label{temporal logics for FO}

In this section, we introduce and discuss two tree temporal logics that are
known to be equivalent to \FO over infinite trees.
Although the equivalence results are not new, we here include a rather detailed
section dedicated to these two logics, with the aim of recasting them in a
syntax closer to the community standard and of presenting brand new results
about them (including a new normal form) to shed light on their behaviour.
Doing this has already helped us better understand the behaviour of \FO over
infinite trees and it might also be useful to others working on the same
topic.

\subsection{Polarised \texorpdfstring{\pmb{\PCTL}}{PastCTL}}
The first formalism we present is called \emph{Polarised \PCTL}, denoted
\PolPCTL. The logic was already introduced in \cite{Sch92a},
under a different notation, namely $\{\mathcal{U}, \mathcal{S}, \mathcal{X}_k\}$, and was proved to be as expressive as
\FO over infinite trees under initial equivalence \cite[Theorem 4.5]{Sch92a}, i.e., equivalence of formulae at the root.
Here we present it in a revised syntax, more similar to the other logics we will work with.

\begin{defi} [\PolPCTL syntax] \label{def:syntaxPolPCTL}
Formulae of \PolPCTL are generated by the following grammar, where $p$ ranges over a set $\AP$ of propositional letters:
%
%
\begin{align*}
\varphi ~::= &~
  p
  ~\bigm|~
  \neg \varphi
  ~\bigm|~
  \varphi \lor \varphi
  ~\bigm|~
  {\D^n\varphi}
  ~\bigm|~
  \E \X \varphi
  ~\bigm|~
  \E \varphi \U \varphi
  ~\bigm|~
  \E \Y \varphi
  ~\bigm|~
  \E \varphi \s \varphi.
\end{align*}
\end{defi}

As is standard in \CTL (or variants thereof), each temporal operator is immediately paired with a path quantifier, so these pairs can be treated as state modalities, making a separate category of path formulae unnecessary.
We make a few comments about the adopted syntax, compared to the original one
from \cite{Sch92a}.
In his work, Schlingloff uses a strict version of the until operator, which he denoted $\mathcal{U}$, with an implicit
path quantification (see, e.g., \cite[Example 2.7]{Sch92a}). Every formula $\psi_1 \mathcal{U} \psi_2$
of the original formalism can be translated to an equivalent \PolPCTL (and thus \PCTL) formula $\E\X \: \E( \psi_1 \U \psi_2)$.
Conversely, \PolPCTL formulae $\E\X\psi$ and $\E\psi_1\U\psi_2$ can be rewritten
using Schlingloff's until operator $\mathcal{U}$, respectively as $\bot
\mathcal{U} \psi$ and $\psi_2 \lor (\psi_1 \land \psi_1 \mathcal{U} \psi_2)$.
A similar (back and forth) translation exists between operators $\E\Y$ (resp.,
$\E\U$) and $\mathcal{S}$.
As far as the counting modality is concerned, Schlingloff's operator
$\mathcal{X}_k$
is simply renamed to $\D^n$ in \PolPCTL (and \PCTL), following a recent practice
\cite{moller2003counting}.
Finally, to maintain uniformity with both the syntax and semantics of \PCTLs,
we require the past temporal operators $\Y$ and $\s$ to be preceded by a path
quantifier as well.
At an intuitive level, this may appear redundant, since nodes of a tree admit
only one past.
Formally, however, it is necessary under our convention of always considering
initial paths.
Overall, this is how we can move between Schlingloff's original formalism and
the syntax of \PCTL.


However, there is an important difference between $\PCTL$ and its polarised
fragment $\PolPCTL$, which concerns the
possibility of obtaining formulae that are initially equivalent to $\E(\psi_1 \R \psi_2)$ and $\A (\psi_1 \U \psi_2)$.
For example, these formulae cannot be directly generated by the grammar of $\PolPCTL$, since the grammar enforces
a ``polarisation'' condition that requires the operator $\E$ (resp., $\A$) to
always appear in front of $\X$ or $\U$ (resp., $\X$ or $\R$).
It is also easy to see that the properties defined by $\E(\psi_1 \R \psi_2)$ and
$\A (\psi_1 \U \psi_2)$ cannot, in fact, be expressed in $\PolCTL$, i.e., ~the
fragment of $\PolPCTL$ in which past temporal operators are prohibited.
Yet, in principle, this does not exclude the possibility that those properties
are definable in $\PolPCTL$, for example using past operators.
However, Item (2) of the theorem below excludes this possibility.
In its turn, the presence of past operators raises additional questions that are
concerned with the expressive power of these formalisms.
We have already seen that adding past operators to a logical formalism does not
always increase its expressive power: this is the case for \CTLs, for instance,
whose past operators can be removed while preserving initial equivalence
(Proposition \ref{PCTLs-expressiveness}).
Here, we will see an opposite situation: \PolCTL is strictly less expressive
than \PolPCTL.


\begin{thm} \label{polctl expressiveness}
Clearly, we have $\PolCTL \subseteq \PolPCTL \subseteq \PCTL$.
The following separation results also hold:
\begin{enumerate}
\item \label{item:polCTL-less-than-pastPolCTL}
  \PolCTL is strictly less expressive than \PolPCTL, i.e., \PolCTL $\subsetneq$ \PolPCTL.
\item \label{item:pastPolCTL-and-CTL-are-incomparable}
  \PolPCTL and \CTL are incomparable.
\item \label{item:pastPolCTL-less-than-pastCTL}
  \PolPCTL is strictly less expressive than \PCTL, i.e., \PolPCTL $\subsetneq$ \PCTL.
\end{enumerate}
\end{thm}

\begin{proof}
All results can be derived from the literature.
The proof of (1) can be derived from \cite[Lemma A.1, Lemma A.2]{laroussinie1995hierarchy}.
Specifically, it is shown that the \CTLs formula $\E(s \lor a\U b)\U r$ is not expressible in \CTL, and consequently neither in \PolCTL,
which is clearly a syntactic fragment of \CTL. Intuitively, the formula states that there exists a path along which $r$ eventually holds and, at every node before that occurrence of $r$, either $s$ holds or $a$ holds until $b$ is reached along the same path. The latter requirement may be fulfilled by an occurrence of $b$ after $r$ has been reached.
On the other hand, the same formula is expressible in \PCTL, as follows
%
\[
  \begin{array}{lcl}
    r & \lor & \E\F( b \land \E\Y\E(a \s [r \land a \land \neg \E\OO (\neg a \land \neg b \land \E((\neg b) \s (\neg b \land \neg s)))]))\\
    & \vee & \E\F(r \land b \land \E\Y(\neg\E\OO (\neg a \land \neg b \land \E((\neg b)\s(\neg b \land \neg s))))) \\
    & \vee & \E\F(r \land \E\Y (\neg \E\OO (\neg a \land \neg b \land \E((\neg b) \s (\neg b \land \neg s))) \land  \neg \E((\neg b)\s (\neg b \land \neg s)))) .
  \end{array}
\]
Since the above formula can clearly be expressed in \PolPCTL, (1) follows.
Moreover, the existence of a property definable in \PolPCTL but not in \CTL establishes one direction of (2).
To get the other direction, we exploit \cite[Theorem 3]{BBMP23}. There, it is shown that
the \CTL formula $\A\F p$ ($= \A (\top \U p))$ has no \FO equivalent. We recall from \cite[Theorem 4.5]{Sch92a}
that \PolPCTL is expressively equivalent to \FO, and therefore the \CTL formula $\A\F p$ is not expressible in \PolPCTL either.
Finally, because the formula $\A\F p$ is also in \PCTL, (3) holds as well.
\end{proof}

\subsection{\texorpdfstring{\pmb{\CTLs}}{CTL*} over finite paths}
The second logic that we discuss here, \emph{Full Computation Tree Logic} over
finite paths, denoted \CTLsf, has a rather different flavour than \PolPCTL, but
it is also shown to be equivalent to \FO.
In \cite[Proposition 3]{BBMP23}, equivalence of \FO and
\CTLsf over unranked trees\footnote{The same equivalence over binary trees was
already noted in \cite[Section 5]{HT87}} is obtained by a
straightforward adaptation of \cite{moller2003counting}, where it was shown that
\CTLs is equivalent to Monadic Path Logic.
%
%
Syntactically, \CTLsf is identical to \CTLs.
The only difference comes at the semantic level, according to which path
quantification ranges over finite paths, rather than infinite ones.
To make the difference more explicit, we will write $\E^{f}$ instead of $\E$.
Below, we report only those clauses in the semantic definition of \CTLsf that
differ from \CTLs; all remaining clauses are unchanged, except that every path
formula is obviously evaluated with respect to a \emph{finite} path $\pi$.

\begin{defi} [\CTLsf semantics]
     %
  Given a pointed tree $(t,u)$, the satisfaction relation $\models_s$ for a
  \CTLsf state formula $\E^f \psi$ 
  is defined as follows:
    \begin{itemize}
        \item $(t, u) \models_s \E^f \psi$ iff there is a finite path $\pi$ starting at the root and visiting $u$ at position $i$ such that $(t, \pi, i) \models_p \psi$.
    \end{itemize}
  Similarly, given a path-pointed tree $(t, \pi, i)$, the satisfaction relation $\models_p$
  for a \CTLsf path formula $\X \psi$ 
  is defined as follows
    \begin{itemize}
        \item $(t, \pi, i) \models_p \X \psi$ iff $i$ is not the last position of $\pi$ and $(t, \pi, i+1) \models_p \psi$.
    \end{itemize}
\end{defi}

We notice an important consequence of the finite-path semantics:
%
%
the \LTL equivalence $\X \neg\psi \leftrightarrow \neg \X \psi$ no longer holds
over finite-path semantics.
For example, $\neg\X\top$ holds vacuously at the last position of a path,
whereas $\X\neg\top$ does not.
Therefore, it makes sense to introduce a new abbreviation $\tilde{\X}\psi = \neg \X \neg\psi$,
which holds at a position $i$ of a path $\pi$ if, whenever the next position exists, it satisfies $\psi$.
While the finite-path semantics has no particular consequence for the operator $\U$,
it does have for the operator $\R$, and specifically for the unfolding of a formula
$\psi_1 \R \psi_2$, which must be defined as
$\psi_2 \land (\psi_1 \lor \tilde{\X} (\psi_1 \R \psi_2))$ in order to preserve equivalence
(note that this involves the new operator $\tilde{\X}$).
%
%

Apart from these nuances, what is interesting about \CTLsf is that, despite the
apparent difference w.r.t.~\PolPCTL, it turns out to have the same expressive
power as \FO and, hence, as \PolPCTL.
In other words, the \emph{syntactic} restrictions of \PolPCTL are precisely
equivalent to the \emph{semantic} restrictions of \CTLsf.
For example, the next lemma shows that the $\R$ operator becomes redundant
within the scope of $\E^f$.
The final theorem of this subsection is in fact considerably deeper in its
implications, since it provides a syntactic restriction of \CTLs equivalent to
the semantic restriction of \CTLsf.
Although the resulting grammar does not, strictly speaking, define a fragment of
\CTLsf, we will nonetheless refer to it as the \CTLsf normal form.

\begin{lem}
\label{cleaning syntactically CTLsf}
  The following equivalence between \CTLsf path formulae is valid, according to
  the finite-path semantics:
\begin{itemize}
\item
  $\varphi \R \psi \leftrightarrow \alpha_{\varphi \R \psi}$, where
  $\alpha_{\varphi \R \psi} = \psi \U ((\tilde{\X} \bot \vee \varphi) \wedge
  \psi)$.
\end{itemize}
The following full state equivalences between \CTLsf state formulae are valid:
\begin{enumerate}
\item
  $\E^f \tilde{\X} \varphi \leftrightarrow \top$,
\item
  $\phi \leftrightarrow \beta_\phi$, where $\beta_\phi$ is obtained from $\phi$
  by replacing every occurrence of subformulae of the form $\varphi \R \psi$
  with $\alpha_{\varphi \R \psi}$.
\end{enumerate}
\end{lem}

\begin{proof}
The proof of the equivalence for path formulae relies on simple inductive
arguments and the use of the unfolded semantics of the operators $\U$ and $\R$.
We start by proving the left-to-right direction of this equivalence.
Suppose $(t,\pi,i) \models \varphi \R \psi$, for some path-pointed tree
$(t,\pi,i)$, with $\pi$ finite.
By the unfolding of $\varphi \R \psi$, we have $(t,\pi,i) \models \psi \land
(\varphi \lor \tilde{\X}(\varphi \R \psi))$.
We exploit an induction on $|\pi| - i$ to prove that $(t,\pi,i) \models
\alpha_{\varphi\R\psi}$.
The base case is where $i$ is the last position of $\pi$, which implies
that $(t,\pi,i) \models \tilde{\X}\bot$ holds vacuously.
Together with $(t,\pi,i)\models \psi$, this immediately implies $(t,\pi,i)
\models \alpha_{\varphi\R\psi}$.
For the inductive step, we assume that $i$ is not the last position of $\pi$,
and recall that both formulae $\psi$ and $\varphi \lor \tilde{\X}(\varphi \R
\psi)$ hold at position $i$.
We distinguish two subcases, depending on whether $(t,\pi,i) \models \varphi$ or
$(t,\pi,i+1) \models \varphi \R \psi$. In the former subcase, together with
$(t,\pi,i)\models\psi$, we easily get $(t,\pi,i) \models
\alpha_{\varphi\R\psi}$.
In the latter subcase, we exploit the inductive hypothesis to get $(t,\pi,i+1)
\models \alpha_{\varphi\R\psi}$, which paired with $(t,\pi,i)\models\psi$ gives
$(t,\pi,i)\models\alpha_{\varphi\R\psi}$.

For the other direction, suppose $(t,\pi,i) \models \alpha_{\varphi\R\psi}$.
By the unfolding of $\alpha_{\varphi\R\psi}$, we get $(t,\pi,i) \models
((\tilde{\X}\bot \lor \varphi) \land \psi) \lor (\psi \land
\X\alpha_{\varphi\R\psi})$.
Accordingly, we distinguish two cases.
If $(t,\pi,i) \models (\tilde{\X}\bot \lor \varphi) \land \psi$, then $(t,\pi,
i) \models \varphi \land \psi$ or $(t,\pi, i) \models \tilde{\X}\bot \land
\psi$. Either way, this immediately implies a satisfaction relation for
$\varphi\R\psi$.
For the other case, where $(t,\pi, i) \models \psi \land
\X\alpha_{\varphi\R\psi}$, we know that $i$ is not the last position in $\pi$.
We can then exploit the inductive hypothesis on $(t,\pi, i+1) \models
\alpha_{\varphi\R\psi}$ to derive $(t,\pi,i+1) \models \varphi \R \psi$.
Because $(t,\pi,i) \models \psi$ also holds, we conclude $(t,\pi,i)\models
\varphi \R \psi$.

Finally, the equivalences between state formulae in items (1) and (2) are
immediately derived from the the previous equivalence, using the finite-path
semantics.
\end{proof}

Before proving the next lemma, we state a well-known result by Thomas.
Moreover, recall that \LTL formulae can be equally seen as special cases of
\CTLs or \CTLsf path formulae.

\begin{prop}[{\cite[Proposition 2.2]{thomas1988safety}}]
\label{equivalence regex-cosafe}
  \cosafeLTL defines precisely those languages described by $\omega$-regular
  expressions of the form $K\Sigma^\omega$, where $\Sigma$ is the underlying
  alphabet and $K$ is a star-free regular expression over $\Sigma$.
\end{prop}

\begin{lem}
\label{subcase for ctlsf normal form}
For every \LTL formula $\psi$, there is a \cosafeLTL formula $\psi'$ such that,
for every pointed tree $(t,u)$,
\[
  (t,u) \models \E^f \psi
\qquad\text{iff}\qquad
  (t,u) \models \E \psi'
\]
\end{lem}

\begin{proof}
Suppose that $(t,u) \models \E^f \psi$, where $\psi$ is an \LTL formula.
Since $\psi$ is a formula interpreted over finite paths, $\Lang(\psi)$ is a
language of \emph{finite} words.
It is known that \LTL and star-free regular expressions are equivalent
formalisms over finite words (see, e.g., \cite{mcnaughton1971counter}).
Hence, there is a star-free regular expression $K$ such that $\Lang(K) =
\Lang(\psi)$.
We extend this star-free regular expression by appending $\Sigma^\omega$, where
$\Sigma$ is the underlying alphabet, thus obtaining the $\omega$-regular
expression $K\Sigma^\omega$.
By Proposition \ref{equivalence regex-cosafe}, there is a \cosafeLTL formula
$\psi'$ such that $\Lang(K\Sigma^\omega) = \Lang(\psi')$.
Trivially, for every finite path $\pi$ that starts at $u$ and every infinite
path $\pi'$ that has $\pi$ as prefix, we have $\pi\in\Lang(K)$ iff
$\pi'\in\Lang(K\Sigma^\omega)$.
This immediately implies that $(t,u) \models \E \psi'$.
The other direction that from $(t,u) \models \E \psi'$ derives $(t,u) \models
\E^f \psi$ is straightforward.
\end{proof}

Thanks to the above lemma, we can design a \emph{syntactic} restriction of
\CTLs, denoted \PolCTLs, that captures precisely the \emph{semantic} restriction
of \CTLsf.
We report here the grammar for \PolCTLs, by highlighting the differences against
the \CTLs grammar:
\begin{align*}
\varphi ~::= &~
  p
  ~\bigm|~
  \neg \varphi
  ~\bigm|~
  \varphi \lor \varphi
  ~\bigm|~
  {\D^n\varphi}
  ~\bigm|~
  \E \psi
\\
\psi ~::= &~
  \varphi
  ~\bigm|~
  \cancel{\neg \psi}
  ~\bigm|~
  \pmb{\psi \land \psi}
  ~\bigm|~
  \psi \lor \psi
  ~\bigm|~
  \X \psi
  ~\bigm|~
  \psi \U \psi
  ~\bigm|~
  \cancel{\!\psi_1 \R \psi_2\!}
\end{align*}
Note that the path quantifier operator $\E$ of \PolCTLs
is the standard one that quantifies over infinite paths.
Roughly speaking, \PolCTLs restricts the argument of an existential path quantifier to be \emph{cosafe-like}, and dually the argument
of a universal path quantifier to be \emph{safe-like}: more precisely, after
treating nested state subformulae as atomic propositions,
the resulting path formula can be seen as a \cosafeLTL formula in the existential case and as a \safeLTL formula in the universal case. Furthermore, although the grammar does not allow negation of arbitrary path formulae, negation remains unrestricted at the level of state formulae. Consequently, the logic is closed under negation: whenever $\varphi$ is a state formula, so is $\neg\varphi$.

\begin{thm} \label{normal form CTLsf}
\CTLsf and \PolCTLs are expressively equivalent.
\end{thm}

\begin{proof}
The equivalence-preserving translation from \CTLsf to \PolCTLs is done as usual
by structural induction, where the only interesting case is the translation of a
\CTLsf state formula $\E^f \psi$.
This relies on Lemma \ref{subcase for ctlsf normal form} and on the idea that we
can treat inner state subformulae as atomic propositions.
Specifically, we substitute in $\E^f \psi$ every state subformula $\varphi$ that
begins with $\E^f$ or $\D^n$ with fresh atomic propositions $p_{\varphi}$, thus
obtaining an \LTL formula $\tilde\psi$.
We then apply Lemma \ref{subcase for ctlsf normal form}, which gives a
\cosafeLTL formula $\tilde\psi'$ for which $\E^f \tilde\psi$ and $\E
\tilde\psi'$ are equivalent.
Finally, to get a formula over the original alphabet equivalent to $\E^f \psi$,
we simply replace every atomic proposition $p_{\varphi}$ introduced earlier with
a corresponding \PolCTLs state formula $\varphi'$ equivalent to $\varphi$, which
exists thanks to the inductive hypothesis.

The opposite translation, from \PolCTLs to \CTLsf, also uses structural
induction, where the interesting case is the translation of a \CTLs state
formula of the form $\E \psi$.
As before, we replace in $\psi$ every state subformula $\varphi$ that begins
with $\E$ or $\D$ by a fresh atomic proposition $p_\varphi$.
This results in a \cosafeLTL formula $\tilde\psi$.
Then we observe that a \cosafeLTL formula holds over an infinite path if and
only if it holds over some finite prefix of it.
This means that we can directly claim that $\E^f \tilde\psi$ is equivalent to
$\E \tilde\psi$.
Finally, by replacing the atomic propositions $p_\varphi$ introduced earlier
with corresponding \CTLsf state formulae $\varphi'$ equivalent to $\varphi$,
which exist by the inductive hypothesis, we obtain a \CTLsf formula $\E^f \psi'$
that is equivalent to $\E \psi$.
\end{proof}

For convenience, and specifically to avoid switching between the finite and the
infinite semantics of path quantifiers, we will mostly work with \PolCTLs as a
normal form of \CTLsf.


\subsection{Connections with Weak Chain Logic}
In the previous sections we have seen that \FO on trees is as expressive as
both \PolPCTL and \CTLsf, and that the latter logic enjoys a normal form (Theorem \ref{normal form CTLsf})
that enforces, at the syntactic level, a polarisation property between the types
of path quantifiers and the types of state formulae used in their arguments.
More precisely, the normal form \PolCTLs binds existential path quantification with co-safety properties,
and, dually, universal path quantification with safety properties.
We will now see that, from a broad perspective, this polarisation phenomenon is
not unique, and is shared by other formalisms more expressive than \FO. For example, by translating the previous logics into the \emph{modal $\mu$-calculus},
one obtains formulae in which the least fixpoint operator $\mu$ is paired only with
the $\diamond$ modality (this would allow capturing the \PolPCTL operator $\E\U$),
and similarly the greatest fixpoint operator $\nu$ is paired only with the $\square$ modality.
Interestingly, this polarisation between $\mu$ and $\diamond$ (or dually,
between $\nu$ and $\square$)
was already observed in connection with \emph{Weak Chain Logic}, whose expressive power is
strictly between that of \FO and that of \emph{Weak MSO}.
We shall explain this in more detail, after the preliminary definitions for the
involved logics.

\emph{Monadic Second Order Logic} (\MSO) is the extension of \FO that introduces
monadic second-order variables, enabling quantification over arbitrary sets of elements
of the underlying structure, in our case sets of nodes of trees. We recall that here
trees are infinite, and in particular the quantified sets of nodes can be finite or infinite.
\emph{Weak MSO} (\WMSO) is the restriction of \MSO in which second-order quantification
ranges over finite sets.
Another fragment of \MSO is \emph{Monadic Chain Logic} (\MCL), which
restricts monadic quantification to (finite or infinite) sets of nodes that
are totally ordered by the ancestor relation.
\emph{Weak Chain Logic} (\WCL) further restricts the latter sets to be finite.
Finally, a couple of even more restricted fragments are \emph{Monadic Path Logic} (\MPL)
and \emph{Weak Path Logic} (\WPL), that constrain the quantified sets to be paths and finite paths,
respectively.

To put these logics in context, it is easy to see that \WPL is equivalent to \FO,
since every finite path can be represented by its first and last nodes,
and the ancestor relation is available in \FO.
Similarly, it was shown in \cite[Lemmas 7-8]{moller2003counting} that \MPL is equivalent to
\CTLs.
The relative expressive power of these common fragments of \MSO can then be
summarised as follows:

\begin{prop}[{\cite[Theorem 3.1]{Tho87}}\cite{Tho84}]
  \leavevmode
  \begin{itemize}
    \item \MPL is strictly less expressive than \MCL,
    \item \MCL is strictly less expressive than \MSO,
    \item \FO is strictly less expressive than \WCL,
    \item \WCL is strictly less expressive than \WMSO.
\end{itemize}
\end{prop}

It turns out that \WCL, which is a natural fragment of \WMSO above \FO, enjoys a
polarisation property similar to that of \FO.
To explain this polarisation property, we quickly recall some expressiveness
results from the literature and present them
in the simpler setting of infinite, binary trees.
Note that on binary trees counting modalities are useless, and so they are not
introduced in the logics.
The first important reference is \cite[Theorem 1]{Car15}, where it is
shows that the bisimulation invariant fragment of \WCL is equivalent to
\emph{Propositional Dynamic Logic} (\PDL).
On binary trees, this is the same as saying that \WCL is expressively equivalent
to \PDL.
Moreover, an earlier work \cite[Propositions 3.8-3.9]{CV14}
characterized \PDL (and hence \WCL) as a fragment of the modal $\mu$-calculus,
called \emph{modal $\mu_{ca}$-calculus}.
Its syntax is reported below (for succinctness, we omit the standard semantics).

\begin{defi} [Modal $\mu_{ca}$-calculus syntax]
Formulae of the \emph{modal $\mu_{ca}$-calculus} are generated by the following (infinite) grammar:
\begin{align*}
\varphi &::=\;
  p
  \;\bigm|\;
  \neg \varphi
  \;\bigm|\;
  \varphi \lor \varphi
  \;\bigm|\;
  \diamond \varphi
   \;\bigm|\;
  \mu X. \psi_{\{X\}}
\\[1ex]
\psi_V &::=\;
  X
  \;\bigm|\;
  \varphi
  \;\bigm|\;
  \psi_V \lor \psi_V
  \;\bigm|\;
  \psi_V \land \pmb {\varphi}
  \;\bigm|\;
  \diamond \psi_V
   \;\bigm|\;
  \mu Y. \psi_{V\cup\{Y\}}
\end{align*}
where $p$ ranges over a finite set $\mathrm{AP}$ of atomic propositions and $V$
denotes a finite set of variables (e.g.~$V=\{X\}$).
Note that the formulae $\varphi$ do not contain free variables, and every
conjunction inside a least fixpoint operator involves at least one such formula
$\varphi$, i.e., conjunctions are linear.
\end{defi}

Here we do not get involved in the technicalities of the definition, but it
should be clear by inspecting the above grammar that it is not possible to
generate a formula that contains a negated $\diamond$ modality (i.e., ~$\square)$
between any variable $X$ and its bounding least fixpoint operator $\mu X$.
Due to the equivalence of modal $\mu_{ca}$-calculus and \WCL over binary trees,
this suggests that monadic quantification restricted to \emph{finite}, totally
ordered sets features a peculiar interaction between fixpoints and modalities,
namely, a sort of polarisation property.



\section{Graded Tree Automata}
In this section we recall the definitions for a general model of automaton
running on infinite, unranked trees, called Graded Tree Automaton, and we
specialise this notion into different classes of automata that we are going to
work with later.
The model is based on classical alternating B\"uchi tree automata \cite[Section
2.2]{wilke2001alternating}, enhanced with graded modalities as done in, e.g.,
\cite[Section 3]{kupferman2002complexity}.
For simplicity, we will start from a one-way automaton model, where transitions
propagate states from node to children, and then we will progressively restrict
or enhance the model so as to match the expressive power of the different
fragments of \PCTLs.
Since every tree automaton introduced here is alternating, B\"uchi, and graded,
we will not repeat these properties every time or include them in the names of
the modified models.
In particular, we will drop the term `graded' early on, instead introducing
qualifications for more important semantic conditions.

\subsection{Graded Tree Automata}
\label{graded tree}

Let $X$ be a set.
We denote by $\mathcal{B}^+(X)$ the set of \emph{positive Boolean formulae} over
$X$, which are obtained from conjunctions and disjunctions of the atomic
formulae $\top$ (true), $\bot$ (false), and $x$, for all $x\in X$.
A \emph{Graded Tree Automaton} (\GTA) is a tuple $\mathcal{A} = \langle Q,
\Sigma, \delta, q_I, F\rangle$, where $Q$ is a finite set of states, $\Sigma$ is
a finite alphabet, $\delta: Q \times \Sigma \rightarrow
\mathcal{B}^+(\{\diamond_k, \square_k\}_{k \in \mathbb N} \times Q)$ is the
transition function, $q_I$ is the initial state, and $F \subseteq Q$ is the set
of accepting states.
The function $\delta$ describes a one-way (or downward) transition that, based
on the current state $q$ of the automaton and the symbol $a$ read at the current
node of the input tree, sends a set of target states to the possible children of
the current node, so as to satisfy the positive Boolean formula $\delta(q,a)$,
where each atom $(\diamond_k,q')$ (resp., $(\square_k, q')$) of the formula is
considered as satisfied if state $q'$ is sent to at least $k$ children (resp.,
if $q'$ is sent to all but $k-1$ children).
In particular, the atom $(\diamond_1,q')$ (often abbreviated $(\diamond,q')$)
requires $q'$ to be sent to at least one child, and can be used, for instance,
to check a formula like $\E\X\phi$; similarly, $(\square_1,q')$ (often
abbreviated $(\square,q')$) requires $q'$ to be sent to all children, and can be
used to check a formula like $\A\X\phi$.
Of course, the atoms $(\diamond_k,q')$ and $(\square_k,q')$ for arbitrary $k$
are often used to simulate the counting modalities.
Given a \GTA $\mathcal{A} = \langle Q, \Sigma, \delta, q_I, F\rangle$ and a
state $q \in Q$, we denote by $\mathcal{A}^q$ the tuple $\langle Q, \Sigma,
\delta, q, F\rangle$, that is, the automaton obtained from $\mathcal{A}$ using
$q$ instead of $q_I$ as the initial state.
To ease some translations later, we also introduce a variant of the previous
notation: given a positive Boolean formula $\theta \in
\mathcal{B}^+(\{\diamond_k, \square_k\} \times Q)$, we denote by
$\mathcal{A}^\theta$ the automaton obtained from $\mathcal{A}$ by adding a fresh
initial state $q'_I$ and by defining $\delta(q'_I, \sigma) = \theta$ for every
$\sigma \in \Sigma$.

%
Let $\mathcal{A} = \langle Q, \Sigma, \delta, q_I, F\rangle$ be a GTA
and $t$ an input tree, labeled over $\Sigma$.
To define a run of $\mathcal{A}$ on $t$, we reuse the notion of tree, now
labeled over the infinite alphabet $Q\times\Dom(t)$.
Note that such a tree labels every node $u$ by a pair $(q,v)$, where $q$
is a state of $\mathcal{A}$ and $v$ is a node in the input tree $t$.
Intuitively, the labeling $(q,v)$ represents a copy of the
automaton $\mathcal{A}$ reading node $v$ with state $q$.
Let $r$ be a tree labeled over $Q\times\Dom(t)$ (i.e.~a potential run of
$\mathcal{A}$).
We formally define the satisfaction relation between a node $u$ of $r$
and a positive Boolean formula of $\mathcal{B}^+(\{\diamond_k, \square_k\}_{k \in \mathbb N} \times Q)$,
inductively as follows:
\begin{itemize}
    \item $u \models (\diamond_k, q')$ if $u$ has a label of the form $(q,v)$ and at least $k$ children $u_1,\dots,u_k$
          with labels $(q',v_1),\dots,(q',v_k)$,
          where $v_1,\dots,v_k$ are distinct children of the node $v$ in $t$;
    \item $u \models (\square_k, q')$ if $u$ has a label of the form $(q,v)$, and for all but $k-1$ children $v'$ of $v$ in $t$,
          there is a child $u'$ of $u$ in $r$ with label $(q',v')$;
    \item $u \models \phi_1 \land \phi_2$ if $u\models \phi_1$ and $u\models \phi_2$;
    \item $u \models \phi_1 \lor \phi_2$ if $u\models \phi_1$ or $u\models \phi_2$.
\end{itemize}

We are now ready to define the notion of run of $\mathcal{A}$ on $t$.
As a matter of fact, because later we will enhance the model with two-way transitions,
which can propagate states from a node to some children and possibly to the parent (if exists),
it is convenient to define our runs starting from an arbitrary node of the input tree, rather
than from the root.
Let $s \in \Dom(t)$.
A $(Q \times \Dom(t))$-tree $r$ is a \emph{run} of $\mathcal{A}$ on $t$
starting from $s$ (\emph{$(t,s)$-run} for short) if it satisfies the following conditions:
\begin{enumerate}
\item the root of $r$ is labeled by $(q_I,s)$, i.e.~the pair consisting of the
      initial state of $\mathcal{A}$ and the starting node in $t$;
\item for every node $u$ of $r$ with label $(q,v)$, if $\sigma$ is the label of $v$ in $t$,
      then $u \models \delta(q,\sigma)$.
\end{enumerate}
Concerning the above condition (2), we observe that the definition of the satisfaction
relation $\models$ between nodes and positive Boolean formulae requires the
existence
of a minimal set of children of each node $u$ of a run. On the other hand, it also
allows a node $u$ to contain spurious children that are not strictly needed to
satisfy condition (2). In particular, the property of being
a run of $t$ is preserved under the addition of subtrees that are themselves runs on
$t$, possibly starting at different nodes.
This turns out to be useful in order to avoid dealing with runs with leaves,
and thus respecting our adopted definition of tree: indeed, even in the presence
of a transition of the form $\delta(q,\sigma)=\top$, for which it is tempting to allow
any $(q,v)$-labeled node $u$ of a run to be a leaf, we can still assume that
there
is an infinite path from $u$, e.g.~one that repeats copies of $u$ itself
(note every such copy of $u$ would still satisfy condition (2) above).
In addition, we highlight an equivalent characterisation of runs, which
will come handy later.
\begin{remark} \label{rem:coinductiverun}
  A tree $r$ is a $(t,s)$-run (resp., an accepting $(t,s)$-run) of a \GTA $\mathcal A$
  if and only if
  \begin{enumerate}
  \item $\epsilon_r \models \delta(q_I, \sigma)$, where $\epsilon_r$ is the
    root of $r$, with label $(q_I,s)$, and $\sigma$ is the label of $s$ in $t$,
    and
  \item for every child $u$ of $\epsilon_r$ with label $(q, v)$, the subtree
        of $r$ rooted at $u$ is a $(t,v)$-run of $\mathcal A^q$.
  \end{enumerate}
\end{remark}

%

%
A run on $t$ is called \emph{initial} when its starting node is the root of $t$.
It is called \emph{accepting} when all of its infinite paths satisfy the B\"uchi
condition, namely, every infinite path has infinitely many nodes labeled by
some
state in $F$.
We say that $t$ is \emph{accepted} by $\mathcal{A}$ if $\mathcal{A}$ admits an
initial, accepting run on $t$.
We denote by $\Lang(\mathcal{A})$ the set of trees accepted by $\mathcal{A}$.
We observe that the property of being an accepting run of $t$ is also preserved
under the addition of subtrees that are themselves accepting runs on $t$,
possibly starting at different nodes.
In a similar way as discussed earlier, this is useful for avoiding accepting
runs with leaves.
More precisely, we assume, without loss of generality, that every automaton
admits at least one accepting run $r_\top$ (not necessarily initial), and that
every $(q,v)$-labeled node of a run that satisfies $\delta(q,\sigma)=\top$ has
an occurrence of $r_\top$ as a subtree.

\subsection{Hesitant Tree Automata}
The class of \emph{Hesitant Tree Automata} was introduced in \cite[Section
5.1]{KVW00} and here we adapt it for our purposes.
A \emph{Hesitant Tree Automaton} (\HTA) is defined as a suitable restriction of
a \GTA $\mathcal{A} = \langle Q, \Sigma, \delta, q_I, F\rangle$.
First, it must admit a partition of the state set $Q$ into disjoint non-empty
sets $Q_1, \ldots, Q_n$, that we call \emph{components}, and a total order
$\leq$ on these components such that, for every $q \in Q_i$ and $q'\in Q_j$, if
$q'$ appears in some atom of $\delta(q, \sigma)$ for some $\sigma$, then $j \le
i$.
Intuitively, transitions can only stay in a component or move to a lower
component. For brevity, we call \emph{$Q_i$-atom} any pair of the form
$(\diamond_k,q)$ or $(\square_k,q)$ whose state $q$ belongs to $Q_i$.
In addition, each component $Q_i$ must be of precisely one of the following
three types:
\begin{itemize}
\item
  \emph{transient}, when there is no $Q_i$-atom in $\delta(q,\sigma)$, for any
  $q\in Q_i$ and $\sigma\in\Sigma$,
\item
  \emph{existential}, when, for all $q\in Q_i$ and $\sigma\in\Sigma$, all
  $Q_i$-atoms in $\delta(q,\sigma)$ are of the form $(\diamond,q')$, and there
  is at most one $Q_i$-atom in each clause of the disjunctive normal form of
  $\delta(q,\sigma)$;
\item
  \emph{universal}, when, for all $q\in Q_i$ and $\sigma\in\Sigma$, all
  $Q_i$-atoms in $\delta(q,\sigma)$ are of the form $(\square,q')$, and there
  is at most one $Q_i$-atom in each clause of the conjunctive normal form of
  $\delta(q,\sigma)$.
\end{itemize}
Note that in the above definition, all counting modalities $\diamond_k$
(resp.~$\square_k$) associated with atoms that stay in the same existential
(resp.~universal) component have cardinality $k=1$.
This intuitively means that a \HTA can traverse at most one node-to-child
relationship while remaining in the same existential component. Moreover, by the
above restrictions, every infinite path of a run of a \HTA eventually gets stuck
in an existential or universal component.
All together, these properties let us capture the semantics of existential and
universal path quantification of tree temporal logics by allowing the \HTA to
guess and check a single path in the input tree while remaining in the same
component.


\subsection{Polarised Hesitant Tree Automata}
%
%
A \emph{Polarised Hesitant Tree Automaton} (\PHTA) is an \HTA where every state
$q$ of an existential component is such that $q \notin F$ and, dually, every
state $q$ of a universal component is such that $q \in F$ (for states of
transient components, it does not matter whether they belong or not to $F$, as
they can occur at most once along an infinite path of a run).
The \emph{polarised} condition, originally introduced in the conference
version of this work \cite{benerecetti2025automaton}, is a stronger restriction
than the classical \emph{weak} condition \cite{muller1986alternating}.
The latter only requires each component of an automaton to contain exclusively
accepting states or exclusively rejecting states, without imposing any
constraint on the type of component or on the structural properties of its
transitions.
As the name suggests, the polarised restriction simulates the polarisation
phenomenon observed in 
\PolCTLs.



\section{Automaton-based characterisation of \texorpdfstring{\PolPCTL}{PCTL+-}}
\label{sec:automata-for-pctlplusminus}

In this section, we first introduce a further restriction of \PHTA, and then an
enhancement of their head movement capabilities. We finally show that the
resulting class of automata, which we call \emph{two-way linear \PHTA}, is
equivalent to \PolPCTL.

The restriction is simple: we require components of the hesitant partition of a \PHTA to be \emph{singletons}. We call such an \PHTA \emph{linear}.

The head movement enhancement is instead less trivial. We allow for \emph{two-way} head movements along the input tree, meaning that the automaton can decide to visit the parent of the currently read node, in addition to its children.
This requires an adaptation of the transition function $\delta$.
Formally, a \emph{two-way \PHTA} is a tuple $\langle Q, \Sigma, \delta, q_I, F
\rangle$, where $Q$, $\Sigma$, $q_I$, and $F$ are, as in the definition of \GTA,
a finite set of states, a finite alphabet, the initial state, and the set of
accepting states, respectively.
Instead, the transition function $\delta: Q \times \Sigma \times \{\isroot,
\isnotroot\}
\rightarrow \mathcal{B}^+(\{\diamond_k, \square_k, \uptran \} \times Q)$
presents two notable differences.
On the one hand, the transition function imposes that the evolution of a run
also depends on whether the head is at the root of the input or not, besides the
symbol read by the head and the current state.
On the other hand, to model the ability of moving the head to the parent of the
current node, new atoms of the form $(\uptran, q)$ are now allowed.
Notice that, since every node of the tree (except for the root) has precisely
one parent, the distinction between existential and universal atoms becomes
immaterial.

We define the semantics of an atom of the form $(\uptran, q)$ as follows, with
respect to a node $u$ labeled with $(q, s)$, for some $q \in Q$ and $s \in
\Dom(t)$ and a $\Sigma$-tree $t$:

\begin{itemize}
\item $u \models (\uptran, q')$ if $u$ has at least a child with label $(q',s')$
  where $s'$ is the parent of $s$ in $t$.
\end{itemize}

We consequently adapt the definition of $(t,s)$-run for the case of a two-way linear \PHTA $\mathcal{A}$.
A $(Q \times \Dom(t))$-tree $r$ is a \emph{run} of $\mathcal{A}$ on $t$
starting from $s$ (\emph{$(t,s)$-run} for short) if it satisfies the following conditions:
\begin{enumerate}
\item the root of $r$ is labeled by $(q_I,s)$, i.e.~the pair consisting of the
  initial state of $\mathcal{A}$ and the starting node in $t$;
\item for every node $u$ of $r$ with label $(q,v)$, it holds that
  $u \models \delta(q,\sigma,\rootvar)$, where $\sigma$ is the label of $v$ in
  $t$ and $\rootvar=\isroot$ (resp., $\rootvar=\isnotroot$) if $v$ is (resp., is
  not) the root of $t$.
\end{enumerate}
Remark~\ref{rem:coinductiverun} is adapted to the new context accordingly. 
Due to the introduction of atoms of the form $(\uptran, q)$, we need to
introduce a new hesitant type.
We say that a component $\{ q \}$ in the hesitant partition of a \PHTA automaton
$\mathcal{A}$ is:
\begin{itemize}
\item \emph{upward}, when, for all $\sigma\in\Sigma$, $q$ can appear in
  $\delta(q,\sigma)$ only in the form $(\uptran,q)$.
\end{itemize}

Notice that the distinction between accepting and rejecting states within upward components is immaterial, since, by construction, a run cannot get stuck in an upward component, as this would require a node to have an infinite sequence of ancestors.
We claim that the class of \emph{two-way linear \PHTA} is as expressive as
\PolPCTL, and consequently as expressive as \FO.
This section is devoted to proving such a claim.

\begin{thm} \label{PHTA-POLPCTL equivalence}
    Given a two-way linear \PHTA $\mathcal{A} = \langle Q, \Sigma, \delta, q_I, F \rangle$ over $\Sigma = 2^\mathrm{AP}$, there is a \PolPCTL formula $\varphi_\mathcal{A}$ such that $\Lang(\mathcal{A}) = \Lang(\varphi_\mathcal{A})$. 
\end{thm}
\begin{proof}
  Working with a state-set partitioned into singletons, we follow the approach of
  \cite[Section 7]{loding2000alternating} and \cite[Section
  3]{boker2018automaton}, who dealt with a similar problem.

  Since $\mathcal{A}$ is linear, some of the restrictions imposed on the
  hesitant types become redundant.
  In particular, if $\{ q \}$ is an existential component, then $\delta(q,
  \sigma, \rootvar)$ contains $q$ only in atoms of the form $(\diamond, q)$, for
  all $\sigma \in \Sigma$ and $\rootvar \in \{ \isroot, \isnotroot \}$;
  therefore, the requirement that $q$ occurs in at most one atom in each clause
  of the disjunctive normal form of $\delta(q, \sigma, \rootvar)$ becomes
  redundant, since $(\diamond, q_i) \land (\diamond, q_i)$ is equivalent to
  $(\diamond, q_i)$.
  Similarly, if $\{ q \}$ is a universal component, then $q$ only occurs in
  atoms of the form $(\square, q)$; therefore, it is of no use to impose the
  requirement that $q$ occurs in at most one atom in each clause of the
  conjunctive normal form of $\delta(q, \sigma, \rootvar)$.
  Consequently, we can assume, without loss of generality, that
  $\delta(q,\sigma, \rootvar)$ contains \emph{at most one occurrence} of an atom
  in which $q$ appears, because of the following:
  \begin{itemize}

  \item if $q$ is transient, then $q$ never appears in $\delta(q, \sigma,
    \rootvar)$ by construction;

  \item if $q$ is existential, then the disjunctive normal form of
    $\delta(q, \sigma, \rootvar)$ is (equivalent to) a formula of the form
    $((\diamond,q) \land \alpha_1) \lor ... \lor ((\diamond,q) \land \alpha_k)
    \lor \alpha_{k+1} ... \lor \alpha_m$, where, for each $i \in [1,m]$,
    $\alpha_i$ is a conjunction of atoms in which $q$ does not appear.
    This formula can be rewritten as $((\diamond, q) \land (\alpha_1 \lor
    ... \lor \alpha_k)) \lor (\alpha_{k+1} \lor ... \lor \alpha_m)$;

  \item if $q$ is universal, then the disjunctive normal form of $\delta(q,
    \sigma, \rootvar)$ is (equivalent to) a formula of the form $((\square,q)
    \land \alpha_1) \lor ... \lor ((\square,q) \land \alpha_k) \lor \alpha_{k+1}
    ... \lor \alpha_m$, where, for each $i \in [1,m]$, $\alpha_i$ is a
    conjunction of atoms in which $q$ does not appear.
    This formula can be rewritten as ($(\square, q) \land (\alpha_1 \lor
    ... \lor \alpha_k)) \lor (\alpha_{k+1} \lor ... \lor \alpha_m)$;

  \item if $q$ is upward, then one can apply to the disjunctive normal form of
    $\delta(q, \sigma, \rootvar)$ the same rewriting described for the
    existential case, and obtain a formula in the form $((\Uparrow, q) \land
    (\alpha_1 \lor ... \lor \alpha_k)) \lor (\alpha_{k+1} \lor ... \lor
    \alpha_m)$, where, for each $i \in [1,m]$, $\alpha_i$ is a conjunction of
    atoms in which $q$ does not appear.
  \end{itemize}
  Thus, $\delta(q, \sigma, \rootvar)$ is in a very precise form:
\[
      \delta(q, \sigma, \rootvar)=
  \begin{cases}
    \alpha'_{q,\sigma, \rootvar} & \text{if } \{q\} \text{ is transient} \ \\
    ((\diamond, q) \land \alpha_{q,\sigma, \rootvar}) \lor \alpha'_{q,\sigma,
    \rootvar} & \text{if } \{q\} \text{ is existential} \\

    ((\square, q) \land \alpha_{q,\sigma, \rootvar}) \lor \alpha'_{q,\sigma,
    \rootvar} & \text{if } \{q\} \text{ is universal} \\
    ((\uptran, q) \land \alpha_{q,\sigma, \rootvar}) \lor \alpha'_{q,\sigma,
    \rootvar} & \text{if } \{q\} \text{ is upward}
  \end{cases}
\]
where $\alpha_{q,\sigma, \rootvar}$ and $\alpha'_{q,\sigma, \rootvar}$ are
formulae in disjunctive normal form not containing atoms in which $q$ occurs:
they only contain states belonging to lower partitions.
It is important to observe that if a $(q,s)$-labeled node of a run satisfies
$\alpha'_{q,\sigma,\rootvar}$, where $\sigma$ is the label of node $s$ in the
input tree $t$ and $\rootvar = \isroot$ (resp., $\rootvar = \isnotroot$) if $s$
is (resp., is not) the root of $t$, then the automaton can leave (and never go
back to) the component $\{q\}$ when going down along the run-tree; conversely,
the sole satisfaction of $\alpha_{q,\sigma, \rootvar}$ (i.e., the node satisfies
$\alpha_{q,\sigma, \rootvar}$ but not $\alpha'_{q,\sigma,\rootvar}$) implies
that the automaton must remain in the component.

  \medskip
  \noindent\textbf{Definition of the translation function $f$.}
  Towards the definition of the translation function $f$ from automata to
  formulae of \PolPCTL, we define the following abbreviation for every $\sigma
  \in \Sigma$
  \[
    \psi_\sigma = \bigwedge_{p \in \sigma} p \land \bigwedge_{p \notin \sigma}
    \neg p.
  \]
  Clearly, $\psi_\sigma$ is true on a node $s$ of an input tree $t$ if and only
  if $\sigma$ is the label of $s$.

  Our translation function $f$ associates a formula of \PolPCTL to every state $q
  \in Q$ of an automaton $\mathcal {A}$ and every formula $\theta \in
  \mathcal{B}^+(\{\diamond_k, \square_k, \uptran\} \times Q)$.
  Formally, $f$ is defined as:

  \begin{itemize}

  \item for all $\theta \in
    \mathcal{B}^+(\{\diamond_k, \square_k, \uptran\} \times Q)$,

    \[f(\theta)=
      \begin{cases}
        \top & \text{if } \theta = \top \\
        \bot & \text{if } \theta = \bot \\
        f(\theta_1) \land f(\theta_2) & \text{if } \theta = \theta_1 \land
                                        \theta_2 \\
        f(\theta_1) \lor f(\theta_2) & \text{if } \theta = \theta_1 \lor
                                       \theta_2 \\
        \D^k f(q) & \text{if } \theta = (\diamond_k,q) \\
        \neg\D^k\neg f(q) & \text{if } \theta = (\square_k,q) \\
        \E \Y f(q) & \text{if } \theta = (\uptran,q)
      \end{cases}
    \]

  \item for all $q \in Q$,

    \[f(q)=
      \begin{cases}
        \beta'_{q} & \text{if } \{q\} \text{ is transient} \\
        \E(\beta_{q} \U \beta'_{q}) & \text{if } \{q\} \text{ is existential} \\
        \A(\beta_{q} \W \beta'_{q}) & \text{if } \{q\} \text{ is universal} \\
        \E(\beta_{q} \s \beta'_{q}) & \text{if } \{q\} \text{ is upward}
      \end{cases}
    \]
    where
    \begin{align*}
      \beta_{q} = \quad & \E\Y\top \rightarrow \bigwedge\nolimits_{\sigma
                          \in \Sigma} (\psi_\sigma \rightarrow
                          f(\alpha_{q,\sigma,\isnotroot})) \\
      \wedge \neg & \E\Y\top \rightarrow \bigwedge\nolimits_{\sigma \in \Sigma}
                    (\psi_\sigma \rightarrow f(\alpha_{q,\sigma,\isroot}))
    \end{align*}
    and
    \begin{align*}
      \beta'_{q} = \quad & \E\Y\top \rightarrow \bigwedge\nolimits_{\sigma \in
                           \Sigma} (\psi_\sigma \rightarrow
                           f(\alpha'_{q,\sigma,\isnotroot})) \\
      \wedge \neg & \E\Y\top \rightarrow \bigwedge\nolimits_{\sigma \in
                    \Sigma} (\psi_\sigma \rightarrow
                    f(\alpha'_{q,\sigma,\isroot}))
    \end{align*}
  \end{itemize}
  Note that, for all $\Sigma$-tree $t$ and nodes $s \in \Dom(t)$, we have that
  $(t,s) \models \beta_q$ (resp., $(t,s) \models \beta'_q$) if and only if
  $(t,s) \models f(\alpha_{q,\sigma,\rootvar})$ (resp., $(t,s) \models
  f(\alpha'_{q,\sigma,\rootvar})$), where $\sigma$ is the label of $s$, and
  $\rootvar = \isroot$ if $s$ is the root of $t$, while $\rootvar = \isnotroot$
  if $s$ is not the root of $t$.
  Additionally, observe that formula $\A(\beta_{q} \W \beta'_{q})$ can be
  rewritten as $\neg\E(\neg \beta'_{q} \U (\neg\beta_{q} \land \neg
  \beta'_{q}))$, to match the syntactic definition of \PolPCTL given in
  Definition~\ref{def:syntaxPolPCTL}.
  Finally, the translation of automaton $\mathcal{A}$ is the formula
  $\varphi_\mathcal{A} = f(q_I)$.

  Before showing soundness of our translation, we show that the above definition
  is based on a well-founded induction.
  Let $\rankq = i$ for all $q \in Q_i$ and $\rank\theta$ be the highest rank
  \rankq associated with states $q$ occurring in some atom of formula $\theta$,
  i.e., $\rank\theta = \max{\{ \rankq \mid q \text{ occurs in some atom of }
    \theta \}}$.
  Consider the strict partial order $\prec$ (based on function \ranksymbol) over
  the domain of $f$, namely $Q \cup \mathcal{B}^+(\{\diamond_k, \square_k,
  \uptran\} \times Q)$, defined as follows:

  \begin{itemize}
  \item for all $q, q' \in Q$, $q \prec q'$ if and only if $\rankq <
    \rankqprime$;

  \item for all $\theta, \theta' \in \mathcal{B}^+(\{\diamond_k, \square_k,
    \uptran\} \times Q)$, $\theta \prec \theta'$ if and only if either
    $\ranktheta < \rankthetaprime$ or $\ranktheta = \rankthetaprime$ and
    $\theta$ is a proper subformula of $\theta'$;

  \item for all $q \in Q$ and $\theta \in \mathcal{B}^+(\{\diamond_k, \square_k,
    \uptran\} \times Q)$
    \begin{itemize}
    \item $q \prec \theta$ if and only if $\rankq \leq \ranktheta$
      and
    \item $\theta \prec q$ if and only if $\ranktheta < \rankq$.
    \end{itemize}

  \end{itemize}
  Clearly, $\prec$ is a well-founded strict partial order.
  Consequently, $f$ is defined by well-founded recursion on $\prec$, since all
  recursive calls are to $\prec$-smaller arguments.
  In particular, observe that $\alpha_{q,\sigma,\rootvar} \prec q$ and
  $\alpha'_{q,\sigma,\rootvar} \prec q$, for all $q \in Q$, $\sigma \in \Sigma$,
  and $\rootvar \in \{ \isroot, \isnotroot \}$.

   \medskip
   \noindent\textbf{Correctness of the translation
     ($\Lang(\varphi_{\mathcal{A}}) \subseteq \Lang(\mathcal{A})$).}
   Let $t$ be a $\Sigma$-tree such that $t \models \varphi_{\mathcal{A}}$
   (recall that $\varphi_{\mathcal{A}}=f(q_I)$).
   We build an initial, accepting run of $\mathcal A$ on $t$.
   To this end, we prove, by induction on the partial order $\prec$ defined
   above, that for all $s \in \Dom(t)$ and $x \in Q \cup
   \mathcal{B}^+(\{\diamond_k, \square_k, \uptran\} \times Q)$ (i.e., $x$ is
   either a state or a formula expressing a transition of the automaton) such
   that $(t,s) \models f(x)$, it is possible to extend every node $v$, labeled
   by $(q',s)$ for a generic $q' \in Q$ (in fact, the role of $q'$ is
   irrelevant), with finitely many children $v_1, \ldots, v_k$, labeled by,
   respectively, $(q_1, s_1), \ldots, (q_k, s_k)$, where $q_i \in Q$, $s_i \in
   \Dom(t)$, and $(t,s_i) \models f(q_i)$ for all $i \in \{ 1, \ldots, k \}$, so
   that:
   \begin{itemize}
   \item $v \models \delta(x, \sigma, \rootvar)$, if $x$ is a state (i.e., $x
     \in Q$) and
   \item $v \models x$, if $x$ is a formula (i.e., $x \in
     \mathcal{B}^+(\{\diamond_k, \square_k, \uptran\} \times Q)$).
   \end{itemize}
   This claim provides a coinductive definition of a, not necessarily accepting,
   $(t,s)$-run of $\mathcal{A}^q$, for all $s \in \Dom(t)$ and $q \in Q$ such
   that $(t,s) \models f(q)$.
   We first prove the claim for states $q \in Q$; then, we prove it for formulae
   $\theta \in \mathcal{B}^+(\{\diamond_k, \square_k, \uptran\} \times Q)$;
   finally, we show that the resulting run is accepting.

   Let $s \in \Dom(t)$ and $q \in Q$ be such that $(t,s) \models f(q)$.
   Moreover, let $v$ be a node labeled with $(q',s)$, for a generic $q' \in Q$.
   If $(t,s) \models \beta'_q$, then $(t,s) \models
   f(\alpha'_{q,\sigma,\rootvar})$, where $\sigma$ is the label of $s$ and
   $\rootvar = \isroot$ (resp., $\rootvar = \isnotroot$) if $s$ is (resp., is
   not) the root of $t$.
   By inductive hypothesis (since $\alpha'_{q,\sigma,\rootvar} \prec q$), we can
   extend $v$ with children $v_1, \ldots, v_k$, labeled by, respectively, $(q_1,
   s_1), \ldots, (q_k, s_k)$, so that $(t,s_i) \models f(q_i)$ for all $i \in \{
   1, \ldots, k \}$ and $v \models \alpha'_{q,\sigma,\rootvar}$, and thus $v
   \models \delta(q,\sigma,\rootvar)$ as well.
   If, instead, $(t,s) \not\models \beta'_q$, then $(t,s) \models \beta_q$ and
   $\{ q \}$ must be existential, universal, or upward (it cannot be transient).
   Following the same argument as above, from $(t,s) \models \beta_q$ it follows
   that $v$ can be extended with children $v_1, \ldots, v_k$, labeled by,
   respectively, $(q_1, s_1), \ldots, (q_k, s_k)$, so that $(t,s_i) \models
   f(q_i)$ for all $i \in \{ 1, \ldots, k \}$ and $v \models
   \alpha_{q,\sigma,\rootvar}$.
   To conclude that $v \models \delta(q,\sigma,\rootvar)$ as well, we need to
   further extend $v$ with other children $v'_1, \ldots, v'_h$, labeled by,
   respectively, $(q, s'_1), \ldots, (q, s'_h)$, such that $(t,s'_i) \models
   f(q)$ for all $i \in \{ 1, \ldots, h \}$ and:
   \begin{itemize}
   \item $v \models (\diamond, q)$, if $\{ q \}$ is existential;
   \item $v \models (\square, q)$, if $\{ q \}$ is universal;
   \item $v \models (\uptran, q)$, if $\{ q \}$ is upward.
   \end{itemize}
   If $\{ q \}$ is existential, then from the fact that $(t,s) \models f(q)$ and
   $(t,s) \not\models \beta'_q$, it follows that $(t,s) \models \E\X f(q)$.
   Let $s'$ be the child of $s$ such that $(t,s') \models f(q)$.
   We extend $v$ with a node labeled $(q,s')$, and we are done.
   If $\{ q \}$ is universal, then from the fact that $(t,s) \models f(q)$ and
   $(t,s) \not\models \beta'_q$, it follows that $(t,s) \models \A\X f(q)$.
   Then, all children $s'$ of $s$ are such that  $(t,s') \models f(q)$.
   We extend $v$ with a node $v_{s'}$, labeled by $(q,s')$, for every child $s'$
   of $s$, and we are done.
   If $\{ q \}$ is upward, then from the fact that $(t,s) \models f(q)$ and
   $(t,s) \not\models \beta'_q$, it follows that $(t,s) \models \E\Y f(q)$.
   Let $s'$ be the parent of $s$; then, it holds $(t,s') \models f(q)$.
   We extend $v$ with a node labeled $(q,s')$, and we are done.

   Now, let $s \in \Dom(t)$ and $\theta \in \mathcal{B}^+(\{\diamond_k,
   \square_k, \uptran\} \times Q)$ be such that $(t,s) \models f(\theta)$.
   Moreover, let $v$ be a node labeled with $(q',s)$, for a generic $q' \in Q$.
   We omit the cases where $\theta$ is a Boolean constant ($\top$ or $\bot$),
   which are trivial, and the ones where it is a conjunction or a disjunction,
   which are dealt with by standard induction, and we focus on the remaining
   ones.
   If $\theta = (\diamond_k, q)$, then from $(t,s) \models f(\theta)$ it follows
   that there are $k$ distinct children $s_1, \ldots, s_k$ of $s$ such that
   $(t,s_i) \models f(q)$ for all $i \in \{ 1, \ldots, k \}$.
   We extend $v$ with nodes $v_1, \ldots, v_k$, labeled by, respectively,
   $(q,s_1), \ldots, (q,s_k)$, and we are done.
   If $\theta = (\square_k, q)$, then from $(t,s) \models f(\theta)$ it follows
   that all but at most $k-1$ children $s'$ of $s$ are such that $(t,s') \models
   f(q)$.
   We extend $v$ with a node $v_{s'}$, labeled by $(q,s')$, for every such child
   $s'$ of $s$, and we are done.
   If $\theta = (\uptran, q)$, then from $(t,s) \models f(\theta)$ it follows
   that $s$ is not the root of $t$ and its parent $s'$ is such that $(t,s')
   \models f(q)$.
   We extend $v$ with a node $v'$, labeled by $(q,s')$, and we are done.

   This coinductive construction produces a $(t,s)$-run of $\mathcal{A}^q$, for
   all $s \in \Dom(t)$ and $q \in Q$ such that $(t,s) \models f(q)$.
   In particular, it produces an initial run of $\mathcal{A}$ on $t$, as $t
   \models f(q_I)$.
   We now show that these runs are accepting.
   To this end, let $r$ be one such run.
   We have to show that all of its paths eventually get stuck in a universal
   component.
   As already pointed out, a path of a run can neither get stuck in a transient
   nor in an upward component.
   Let $\{ q \}$ be an existential component.
   By construction, if $v$ is a node of $r$ labeled by $(q,s)$ then $q$ occurs
   in at most one of its children label (i.e., $v$ has at most one child labeled
   by $(q,s')$, for some $s'$); moreover, if $(t,s) \models \beta'_{q}$, then
   there is no child of $v$ labeled by $(q,s')$, for any $s'$.
   Now, assume, towards a contradiction, that $r$ features an infinite path $v_0
   v_1 \dots$ such that $v_i$ is labeled by $(q,s_i)$ for all $i \in \mathbb N$.
   Thus, $(t,s_i) \models f(q)$ and $(t,s_i) \not\models \beta'_{q}$, for all $i
   \in \mathbb N$.
   In particular, since $(t,s_0) \models f(q)$, it must be the case that
   $(t,s_k) \models \beta'_{q}$, for some $k \in \mathbb{N}$, which contradicts
   the fact that $(t,s_i) \not\models \beta'_{q}$ holds for all $i \in \mathbb
   N$.
   Therefore, there exists an initial, accepting run of $\mathcal{A}$ on $t$.

  \medskip
  \noindent\textbf{Completeness of the translation ($\Lang(\mathcal{A})
    \subseteq \Lang(\varphi_{\mathcal{A}})$).}
  Let $t \in \Lang(\mathcal{A})$ and $r$ be an initial, accepting run of
  $\mathcal{A}$ on $t$.
  We show that $t \models \varphi_{\mathcal A}$.
  As a preliminary step, we observe that for every node $v \in \Dom(r)$ with
  label $(q,s)$,  we have that   $v \models \alpha_{q,\sigma,\rootvar}$ or $v
  \models \alpha'_{q,\sigma,\rootvar}$, where $\sigma$ is the label of $s$ and
  $\rootvar=\isroot$ (resp., $\rootvar=\isnotroot$) if $s$ is (resp., is not)
  the root of $t$.
  Moreover, if $v \not\models \alpha'_{q,\sigma,\rootvar}$, then $\{ q \}$ is
  not transient and one of the following holds:
  \begin{itemize}
  \item $\{ q \}$ is existential and there are a child $v'$ of $v$ and a child
    $s'$ of $s$ such that the label of $v'$ is $(q,s')$;
  \item $\{ q \}$ is universal and for every child $s'$ of $s$ there is a child
    $v'$ of $v$ whose label is $(q,s')$;
  \item $\{ q \}$ is upward, $s$ is not the root of $t$, and there is a child
    $v'$ of $v$ whose label is $(q,s')$, where $s'$ is the parent of $s$.
  \end{itemize}

  Next, we show that for every node $v \in \Dom(r)$, with label $(q,s)$, if
  $\subtree_v{r}$ (the subtree of $r$ rooted at node $v$) is an accepting
  $(t,s)$-run of $\mathcal{A}^q$, then $(t,s) \models f(q)$.
  To this end, we prove, by induction on the partial order $\prec$ defined
  above, the following claims, for every $s \in \Dom(t)$:
  \begin{enumerate}


  \item \label{item:automataToFormulaCompletenessStatesLocal}
    for every $q \in Q$ and every node $v \in \Dom(r)$ labeled by $(q,s)$, it
    holds that $(t,s) \models \beta_q$ or $(t,s) \models \beta'_q$,

  \item \label{item:automataToFormulaCompletenessStatesFixpoint}
    for every $q \in Q$ and every node $v \in \Dom(r)$ labeled by $(q,s)$, it
    holds that $(t,s) \models f(q)$,

  \item \label{item:automataToFormulaCompletenessFormulas}
    for every $\theta \in \mathcal{B}^+(\{\diamond_k, \square_k, \uptran\}
    \times Q)$ and every node $v \in \Dom(r)$ labeled by $(q,s)$ for some $q \in
    Q$, if $v \models \theta$, then $(t,s) \models f(\theta)$.

  \end{enumerate}

  Let $q \in Q$ and  $v \in \Dom(r)$ be labeled by $(q,s)$.
  Moreover, let $\sigma$ be the label of $s$ and $\rootvar=\isroot$ (resp.,
  $\rootvar=\isnotroot$) if $s$ is (resp., is not) the root of $t$.
  Since $\alpha_{q,\sigma,\rootvar} \prec q$ (resp.,
  $\alpha'_{q,\sigma,\rootvar} \prec q$), by inductive hypothesis
  Claim~\eqref{item:automataToFormulaCompletenessFormulas} holds, meaning that
  $v \models \alpha_{q,\sigma,\rootvar}$ (resp., $v \models
  \alpha'_{q,\sigma,\rootvar}$) implies $(t,s) \models
  f(\alpha_{q,\sigma,\rootvar})$ (resp., $(t,s) \models
  f(\alpha'_{q,\sigma,\rootvar})$).
  Claim~\eqref{item:automataToFormulaCompletenessStatesLocal} then follows from
  the fact that $v \models \alpha_{q,\sigma,\rootvar}$ or $v \models
  \alpha'_{q,\sigma,\rootvar}$.

  If $(t,s) \models \beta'_q$, then
  Claim~\eqref{item:automataToFormulaCompletenessStatesFixpoint} holds
  straightforwardly.
  If, instead, $(t,s) \not \models \beta'_q$, then $\{ q \}$ is not transient.
  If $\{ q \}$ is existential, then consider a maximal path $v_0 v_1 \dots$
  starting from $v$ (i.e., $v_0=v$) such that $v_i$ is labeled by $(q,s_i)$ for
  all $i$ and $s_0 s_1 \dots$ is a path of $t$ starting at $s$ (i.e., $s_0 =
  s$).
  Such a path cannot be infinite since $r$ is accepting and $\{ q \}$ is
  existential.
  Let $v_k$ be its last node.
  By maximality, there is no child of $v$ labeled by $(q,s')$, for any child
  $s'$ of $s_k$.
  Therefore, $v_k \models \alpha'_{q,\sigma,\rootvar}$ and, by
  Claim~\eqref{item:automataToFormulaCompletenessFormulas} (applying the
  inductive hypothesis), $(t,s_k) \models \beta'_{q}$.
  Since, by Claim~\eqref{item:automataToFormulaCompletenessStatesLocal},
  $(t,s_i) \models \beta_{q} \lor \beta'_{q}$ for all $i \in \{ 0, \ldots, k-1
  \}$, we have that $(t,s) \models \E(\beta_{q} \U \beta'_{q})$, and thus
  Claim~\eqref{item:automataToFormulaCompletenessStatesFixpoint} holds.
  If $\{ q \}$ is universal, then we proceed by assuming, towards a
  contradiction, that $(t,s) \not\models f(q)$, that is, $(t,s) \not\models
  \A(\beta_{q} \W \beta'_{q})$.
  Thus, there is a finite path $s_0 s_1 \dots s_k$, with $k>0$, starting from
  $s$ (i.e., $s_0=s$) such that $(t,s_k) \models \neg \beta_q \land \neg
  \beta'_q$ and, for all $i \in \{ 0, \ldots, k-1 \}$, $(t,s_i) \models \beta_q
  \land \neg \beta'_q$.
  From our preliminary observations, there is a path $v_0 v_1 \dots v_k$
  starting from $v$ (i.e., $v_0=v$) such that $v_i$ is labeled by $(q,s_i)$ for
  all $i \in \{ 0, \ldots, k \}$.
  However, the fact that node $v_k$ is labeled by $(q,s_k)$ and $(t,s_k) \models
  \neg \beta_q \land \neg \beta'_q$ is in contradiction with
  Claim~\ref{item:automataToFormulaCompletenessStatesLocal}.
  Therefore $(t,s) \models \A(\beta_{q} \W \beta'_{q})$.
  If $\{ q \}$ is upward, then consider a maximal path $v_0 v_1 \dots$ starting
  from $v$ (i.e., $v_0=v$) such that, for all $i$, we have that $v_i$ is labeled
  by $(q,s_i)$ and $s_{i+1}$ is the parent of $s_i$.
  Such a path cannot be infinite, or $s$ would have infinitely many ancestors.
  Let $v_k$ be the last node of the path.
  By maximality, there is no child of $v$ labeled by $(q,s')$, where $s'$ is the
  parent of $s_k$.
  Therefore, $v_k \models \alpha'_{q,\sigma,\rootvar}$ and, by
  Claim~\eqref{item:automataToFormulaCompletenessFormulas} (applying the
  inductive hypothesis), $(t,s_k) \models \beta'_{q}$.
  Since, by Claim~\eqref{item:automataToFormulaCompletenessStatesLocal},
  $(t,s_i) \models \beta_{q} \lor \beta'_{q}$ for all $i \in \{ 0, \ldots, k-1
  \}$, we have that $(t,s) \models \E(\beta_{q} \s \beta'_{q})$, and thus
  Claim~\eqref{item:automataToFormulaCompletenessStatesFixpoint} holds.

  Finally, to prove Claim~\eqref{item:automataToFormulaCompletenessFormulas},
  let $\theta \in \mathcal{B}^+(\{\diamond_k, \square_k, \uptran\} \times Q)$
  and $v \in \Dom(r)$ be a node labeled by $(q,s)$ such that $v \models \theta$.
  We omit the cases where $\theta$ is a Boolean constant ($\top$ or $\bot$),
  which are trivial, and the ones where it is a conjunction or a disjunction,
  which are dealt with by standard induction, and we focus on the remaining
  ones.
  If $\theta = (\diamond_k, q)$, then from $v \models \theta$ it follows that
  $v$ has at least $k$ children $v_1,\dots,v_k$ with labels
  $(q,s_1),\dots,(q,s_k)$, where $s_1,\dots,s_k$ are distinct children of $s$.
  Since $q \prec \theta$, by inductive hypothesis
  Claim~\eqref{item:automataToFormulaCompletenessStatesFixpoint} holds, meaning
  that $(t,s_i) \models f(q)$ for all $i \in \{ 1, \ldots, k \}$.
  Therefore, $(t,s) \models \D^k f(q)$, which implies
  Claim~\eqref{item:automataToFormulaCompletenessFormulas}.
  If $\theta = (\square_k, q)$, then from $v \models \theta$ it follows that for
  all but at most $k-1$ children $s'$ of $s$ there is a child $v_{s'}$ of $v$
  with label $(q,s')$.
  Since $q \prec \theta$, by inductive hypothesis
  Claim~\eqref{item:automataToFormulaCompletenessStatesFixpoint} holds, meaning
  that $(t,s') \models f(q)$ holds for all but at most $k-1$ children $s'$ of
  $s$.
  Therefore, $(t,s) \models \neg\D^k\neg f(q)$, which implies
  Claim~\eqref{item:automataToFormulaCompletenessFormulas}.
  If $\theta = (\uptran, q)$, then from $v \models \theta$ it follows that $v$
  has at least a child with label $(q,s')$, where $s'$ is the parent of $s$ in
  $t$.
  Since $q \prec \theta$, by inductive hypothesis
  Claim~\eqref{item:automataToFormulaCompletenessStatesFixpoint} holds, meaning
  that $(t,s') \models f(q)$.
  Therefore, $(t,s) \models \E \Y f(q)$, which implies
  Claim~\eqref{item:automataToFormulaCompletenessFormulas}.
  We conclude by observing that, by
  Claim~\eqref{item:automataToFormulaCompletenessStatesFixpoint}, we have that
  $t \models f(q_I)$, that is, $t \models \varphi_{\mathcal{A}}$.
\end{proof}

We can now state the easier direction of the equivalence.

\begin{thm}
    Given a \PolPCTL formula $\varphi$ over a set of atomic propositions $\mathrm{AP}$, there is a two-way linear \PHTA $\mathcal{A}_\varphi$ such that $\Lang(\varphi) = \Lang(\mathcal{A}_\varphi$). 
\end{thm}

\begin{proof}
    We consider \PolPCTL formulae in \emph{negation normal form}, meaning that negation can occur only at the level of atomic propositions. It is easy to see that for any \PolPCTL formula $\psi$, there is an equivalent formula $\psi'$ in negation normal form. We define the abbreviations $\A\tilde{\Y} \psi= \neg\E \Y \neg\psi$, 
    and $\C^k \psi = \neg \D^k \neg \psi$. Recall that the distinction between existential and universal path quantifiers is irrelevant when they precede past temporal operators. We don't introduce an abbreviation for $\neg\E(\neg\psi_1 \s \neg\psi_2)$, since it can always be rewritten as $\E(\psi_2 \s (\psi_2 \land (\A\tilde{\Y}\bot \lor \psi_1)))$. 

    Given a \PolPCTL formula, let \emph{sub}($\psi)$ be the set of its subformulae, defined as usual for atomic propositions and Boolean operators, and as follows for temporal and counting operators:
    \begin{tasks}[label=\textbullet](2)
    \task $sub(\D^k\psi) = \{\psi\}$
    \task $sub(\C^k\psi) = \{\psi\}$
    \task $sub(\E\X\psi) = \{\psi\}$
    \task $sub(\A\X\psi) = \{\psi\}$
    \task $sub(\E\Y\psi) = \{\psi\}$
    \task $sub(\A\tilde{\Y}\psi) = \{\psi\}$
\end{tasks}
    \begin{tasks}[label=\textbullet](2)
    \task $sub(\E(\psi_1 \U \psi_2)) = \{\psi_1, \psi_2, \E(\psi_1 \U \psi_2)\}$
    \task $sub(\A(\psi_1 \R \psi_2)) = \{\psi_1, \psi_2, \A(\psi_1 \R \psi_2)\}$
    \task $sub(\E(\psi_1 \s \psi_2)) = \{\psi_1, \psi_2, \E(\psi_1 \s \psi_2)\}$
    \end{tasks}

  Given a \PolPCTL formula $\varphi$ over the set of atomic propositions $\mathrm{AP}$, we define a two-way linear \PHTA $\mathcal{A}_\varphi = \langle sub(\varphi), 2^\mathrm{AP}, \delta, \varphi, F \rangle$. The state space $sub(\varphi)$ consists of the subformulae of $\varphi$ defined above, with $\varphi$ itself serving as the initial state. The set of accepting states $F$ comprises all states corresponding to subformulae of the form $\A(\psi_1 \R \psi_2)$. Finally, the transition function $\delta: Q \times \Sigma \times \{\isroot,\isnotroot\} \rightarrow \mathcal{B}^+(\{\diamond_k, \square_k, \uptran\} \times Q)$ is defined as follows, for every $\sigma \in \Sigma$ and with $\rho \in \{\isroot, \isnotroot\}$:
    \begin{tasks}[label=\textbullet](2)
    \task $\delta(p, \sigma, \rho) = \top$, if $p \in \sigma$
    \task $\delta(p, \sigma, \rho) = \bot$, if $p \notin \sigma$
    \task $\delta(\neg p, \sigma, \rho) = \top$, if $p \notin \sigma$
    \task $\delta(\neg p, \sigma, \rho) = \bot$, if $p \in \sigma$
    \end{tasks}
     \begin{tasks}[label=\textbullet]
    \task $\delta(\psi_1 \otimes \psi_2, \sigma, \rho) = \delta(\psi_1, \sigma, \rho) \otimes \delta(\psi_2, \sigma, \rho)$, for $\otimes \in \{\land, \lor\}$
    \end{tasks}
    \begin{tasks}[label=\textbullet](2)
    \task $\delta(\D^k\psi, \sigma, \rho) = (\diamond_k, \psi)$
    \task $\delta(\C^k\psi, \sigma, \rho) = (\square_k, \psi)$
    \task $\delta(\E\X\psi, \sigma, \rho) = (\diamond, \psi)$
    \task $\delta(\A\X\psi, \sigma, \rho) = (\square, \psi)$
    \task $\delta(\E\Y \psi, \sigma, \isnotroot) = (\uptran, \psi)$
    \task $\delta(\E\Y \psi, \sigma, \isroot)) = \bot$
    \task $\delta(\A\tilde{\Y} \psi, \sigma, \isnotroot) = (\uptran, \psi)$
    \task $\delta(\A\tilde{\Y} \psi, \sigma, \isroot) = \top$
    \end{tasks}
    \begin{tasks}[label=\textbullet]
    \task $\delta(\E(\psi_1 \U \psi_2), \sigma, \rho) = \delta(\psi_2, \sigma, \rho) \lor (\delta(\psi_1, \sigma, \rho) \land (\diamond, \E(\psi_1 \U \psi_2))$
    \task $\delta(\A(\psi_1 \R \psi_2), \sigma, \rho) = \delta(\psi_2, \sigma, \rho) \land (\delta(\psi_1, \sigma, \rho) \lor (\square, \A(\psi_1 \R \psi_2))$
    \task $\delta(\E(\psi_1 \s \psi_2), \sigma, \rho) = \delta(\psi_2, \sigma, \rho) \lor (\delta(\psi_1, \sigma, \rho) \land (\uptran, \E(\psi_1 \s \psi_2))$
    \end{tasks}
    
    The automaton is two-way. The linear requirement is satisfied by the subformula ordering: a state corresponding to a formula $\psi$ can only transition to itself or to states corresponding to subformulae of $\psi$. Furthermore, the automaton is hesitant, as states of the form $\E(\psi_1 \U \psi_2)$ are existential, $\A(\psi_1 \R \psi_2)$ are universal, $\E(\psi_1 \s \psi_2)$ are upward, and all remaining states are transient. Finally, it is polarised because only the universal states are accepting. The soundness of this construction is routine; it follows from a straightforward adaptation of the soundness proof for the classical \CTL automaton construction (see, e.g., \cite[Theorem 14.7.11]{demri2016temporal})
\end{proof}

This concludes the proof and implies the following.

\begin{thm} \label{cor:twowaypctl}
    Two-way linear \PHTA and \PolPCTL are equivalent formalisms.
\end{thm}

By the equivalence between \FO and \PolPCTL obtained in \cite[Theorem 4.5]{Sch92a}, one can also conclude the following.

\begin{cor}
    Two-way linear \PHTA and \FO are equivalent formalisms.
\end{cor}

This result is quite surprising, since it shows that two-way head movements, in combination with a very severe restriction as linearity, yield \FO expressiveness over trees. This contrasts with the behaviour of \FO over words, where the class of alternating automata expressively equivalent to it is that of \emph{linear} alternating automata \cite[Section 7]{loding2000alternating}, with no need for two-way head movements. In a sense, this result can be seen as further evidence of the ``ill" (or, at least, unpredictable) behaviour of \FO over trees. In the following sections we will keep reasoning, with automaton-based techniques, on \FO expressiveness, showing another automaton class effectively equivalent to this logic, this time with standard one-way head movements. This ``one-way" characterisation turned out to be rather intricate, and not easily derivable by the above result. Indeed, it can be easily seen by a detour through logic and a straightforward adaptation of the above proofs, that linear \PHTA are strictly less expressive than two-way linear \PHTA, and, consequently, of \FO. 

\begin{thm}
    Linear \PHTA and \PolCTL are equivalent formalisms.
\end{thm}

\begin{proof}
    It suffices to adapt in a straightforward way the above proofs, removing every case involving upward states or past temporal operators, and apply the various inductive hypotheses on \PolCTL (resp., linear \PHTA) instead that on its past (resp., two-way) counterpart. 
\end{proof}

By Theorem \ref{polctl expressiveness}(1), this implies the following.

\begin{cor}
    Linear \PHTA are strictly less expressive than two-way linear \PHTA. 
\end{cor}

In the next section we develop some word automata theory, that will be useful in the following to yield the other \FO automaton-based characterisation here presented.



\section{Automaton-based characterisations of \safeLTL and \cosafeLTL} \label{word automata char}

This section is preliminary to the results presented in Section \ref{second
automaton}, which give another automaton-based characterisation for \FO.
Here we present two specific classes of $\omega$-word automata that characterise
respectively \safeLTL and \cosafeLTL.
These latter results in the setting of word languages are, in our opinion,
interesting on their own.

Here we will mostly consider nondeterministic B\"uchi / coB\"uchi automata
running on $\omega$-words, which we present as tuples of the form $\mathcal{A} =
\langle Q , \Sigma, \Delta, q_I, F \rangle$, where $\Delta \subseteq
Q\times\Sigma\times Q$ is the transition relation (we use the capital $\Delta$
to distinguish it from a transition function), $q_I$ is an initial state, and
$F$ is a set of final states.
We recall from the preliminaries that such an automaton is \emph{counter-free}
if, for every state $q \in Q$, every finite word $u \in \Sigma^*$, and every $n
> 0$, if $(q,u^n,q)\in \Delta^*$, then $(q,u,q) \in \Delta^*$, where $\Delta^*$
is the transitive closure of $\Delta$, i.e., the standard extension of the
transition relation to finite words.
We introduce an additional property below.

\begin{defi}
A nondeterministic automaton $\mathcal{A} = \langle Q , \Sigma, \Delta, q_I, F
\rangle$ is \emph{looping} if there is $q_{sink} \in Q$ such that $F = Q\setminus\{q_\sink\}$ and
$(q_\sink,\sigma,q_\sink)\in\Delta$ for all $\sigma
\in \Sigma$.
\end{defi}

Note that the above definitions make no explicit reference to the acceptance
condition of the automaton (be it B\"uchi or coB\"uchi) or to the semantics of
nondeterminism (whether it should be interpreted as an existential or a
universal choice): these definitions make sense in every combination of such
cases.
Below, we shall use the shorthands \UBA and \NCA respectively for a universal
B\"uchi automaton (accepting an $\omega$-word if all runs on it visit infinitely
often a final state) and an existential coB\"uchi automaton (accepting an
$\omega$-word if at least one run on it visits only finitely often the final
states).

We now recall the following automaton-based characterisation of the \LTL
fragments \safeLTL and \cosafeLTL.

\begin{prop}[\cite{boker2022translation}]
  \label{looping equivalence}
  Counter-free looping \emph{deterministic} B\"uchi (resp., coB\"uchi) automata
  define the same class of languages of $\omega$-words as \safeLTL (resp.,
  \cosafeLTL).
\end{prop}

Our alternative characterisation for \safeLTL and \cosafeLTL essentially shows
that one can drop the assumption of determinism, provided that the
interpretation of nondeterministic choices is universal for automata with
B\"uchi acceptance conditions and existential for automata with coB\"uchi
acceptance conditions.
Note that the latter requirement is as expected, since the two automaton models
for \safeLTL and \cosafeLTL need to be duals of each other, meaning that
complementing an automaton in one class gives an automaton in the other class.
This alternative characterisation will turn out useful in the next section.

\begin{lem} \label{lem:UBAcosafe}
  Counter-free looping \UBA and \NCA define the same class of languages of
  $\omega$-words as \safeLTL and \cosafeLTL, respectively.
\end{lem}

\begin{proof}
In view of the duality between \UBA and \NCA, it suffices to prove only the
characterisation of \safeLTL in terms of \UBA.
As a matter of fact, since deterministic B\"uchi automata can be seen as a
special case of \UBA, it suffices to prove that counter-free looping \UBA can
only define \safeLTL properties.

Consider a counter-free looping \UBA $\mathcal{A} = \langle Q, \Sigma, \Delta,
q_I, Q\rangle$.
We show that, by applying a variant of the subset construction, it is possible
to turn $\mathcal{A}$ into an equivalent counter-free looping deterministic B\"uchi
automaton, which by Proposition \ref{looping equivalence} defines a \safeLTL
property.
More precisely, we use the Miyano-Hayashi breakpoint construction, which was
introduced in \cite[Theorem 4.1]{miyano1984alternating} to obtain a \NBA
starting from an alternating B\"uchi automaton.
We briefly recall this construction. While originally introduced to convert alternating B\"uchi automata into \NBA, applying this construction specifically to a \UBA yields a \emph{deterministic} B\"uchi automaton. We briefly review this procedure. Given $\mathcal{A}$, the states of the resulting deterministic automaton $\mathcal{A}'$ are pairs $(S,O)$ satisfying $S \supseteq O$, with the following intended semantics:
\begin{itemize}
\item
  $S$ represents the classical subset construction, and contains the states
  reached along any run induced by the consumed prefix of the input,
\item
  $O$ represents obligations for visiting $F$ infinitely often, namely, it
  contains the states reached along those runs that have not yet visited $F$
  since the last breakpoint (a breakpoint happens whenever $O$ becomes empty).
\end{itemize}
Formally, the transitions of $\mathcal{A}'$ are described by the deterministic
function $\Delta'$ that maps a pair $(S,O)$ and a symbol $\sigma$ to the pair
$(S',O')$ defined by:
\begin{align*}
  S'
&~=~
  \big\{ q'\in Q \::\: \exists{q\in S} ~~ (q,\sigma,q')\in\Delta \big\}
\\
  O'
&~=~
  \begin{cases}
    \big\{ q'\in Q \setminus F \::\: \exists{q\in O} ~~(q,\sigma,q')\in\Delta
    \big\}
  & \text{if $O\neq\emptyset$,}
  \\
    S' \setminus F
  & \text{otherwise.}
  \end{cases}
\end{align*}
The initial state of $\mathcal{A}'$ is the pair $(\{q_I\},\{q_I\}\setminus F)$,
the final states are those pairs $(S,O)$ with $O=\emptyset$, intuitively,
enforcing infinitely many breakpoints.
The constructed deterministic automaton $\mathcal{A}'$, interpreted as a B\"uchi
automaton, is easily shown to be equivalent to $\mathcal{A}$ (for a formal
proof, see \cite[Theorem 4.1]{miyano1984alternating}).

It remains to show that $\mathcal{A}'$ is counter-free and looping. For this, we
exploit the assumption that the initial automaton $\mathcal{A}$ is counter-free
and looping, and we focus our attention on the second component $O$ of the
states of $\mathcal{A}'$.
Since all states in $\mathcal{A}$ but $q_\sink$ are final, the second component
can only be empty or the singleton $\{q_\sink\}$.
Moreover, whenever it becomes the singleton $\{q_\sink\}$, it stays so forever,
since $(q_\sink,\sigma,q_\sink)\in\Delta$ for every $\sigma\in\Sigma$.
This produces potentially many sink states of the form $(S,\{q_\sink\})$ in
$\mathcal{A}'$, none of which is final. Of course, we can merge all these states
into a single macro-state, say $(\bot,\{q_\sink\})$, where the first component
is forgotten (we accordingly modify the transition rules for the pairs where the
second component is $\{q_\sink\}$).
The resulting B\"uchi automaton $\mathcal{A}''$ is still deterministic and
equivalent to $\mathcal{A}$, but now it is also looping: its sink state is
precisely $(\bot,\{q_\sink\})$.
As a matter of fact, this shows that, when the initial automaton is looping, the
breakpoint construction can be simplified.
For similar reasons, precisely because the projection of $\mathcal{A}''$ onto
the second component $O$ is already counter-free, in order to prove that
$\mathcal{A}''$ is counter-free, it suffices to show that the projection onto
the first component $S$ is also counter-free.
For this we observe that the latter projection is a classical subset
construction, and we exploit \cite[Remark 11.14]{diekert2008first}, where it is
shown that the subset construction of a counter-free automaton yields again a
counter-free automaton.
We thus conclude that $\mathcal{A}''$ is counter-free and looping, and hence by
Proposition \ref{looping equivalence} it defines a \safeLTL property, and so
does $\mathcal{A}$.
\end{proof}

\section{Automaton-based characterisation of \texorpdfstring{\CTLs}{CTL*} over finite paths} \label{second automaton}

In this section we present an automaton-based characterisation of the logic \CTLsf;
more precisely, we will provide effective translations between a restricted variant of 
(one-way) \PHTA and the normal form \PolCTLs of \CTLsf. By the equivalence of the latter logic with \FO, 
it follows that our automaton model captures precisely the expressive power of \FO. 

Towards introducing the restrictions on \PHTA that match the expressive power of \CTLsf,
we start by describing a \emph{linearisation} operation on the components of a given \PHTA. 
Intuitively, this will let us view each transformed component as an $\omega$-word automaton, 
thus enabling the use of classical characterisations of \FO-definable word languages.
Unfortunately, this first definition is quite technical.

\begin{defi} \label{Linearization}
Let $\mathcal{A} = \langle Q, \Sigma, \delta, q_I, F\rangle$ be a \PHTA,
$Q_i$ be a non-transient component of it, and $q\in Q_i$.
We define the $\omega$-word automaton $\mathcal{A}^q_\exit = \langle Q_i\uplus\{q_\exit\}, \Sigma \times 2^B, \Delta, q, Q_i \rangle$,
where $q_\exit$ is a fresh state, $B$ is the set of atoms 
with states outside the component $Q_i$,
$q$ is the new initial state, $Q_i$ is the set of final states, and $\Delta$ is the transition relation defined by a
case distinction as follows:
\begin{itemize}
    \item If $Q_i$ is an existential component, 
          then $\Delta$ contains triples of the form $(q',(\sigma,C),q'')$,
          with $\sigma\in\Sigma$ and $C\subseteq B$, 
          satisfying one of the following conditions:
          \begin{enumerate}
            \item $q'=q''=q_\exit$,
            \item $q'\in Q_i$, $q''=q_\exit$, and the disjunctive normal form of $\delta(q',\sigma)$ contains a clause of the form $\bigwedge C$,
            \item $q',q''\in Q_i$ and the disjunctive normal form of $\delta(q',\sigma)$ contains a clause of the form $(\diamond, q'') \land \bigwedge C$.
          \end{enumerate}
          In this case we interpret $\mathcal{A}^q_\exit$ as a nondeterministic coB\"uchi automaton (\NCA),
          which accepts an $\omega$-word whenever some run on it eventually leaves the component $Q_i$.
    %
    \item Dually, if $Q_i$ is a universal component, 
          then $\Delta$ contains triples of the form $(q',(\sigma,C),q'')$,
          with $\sigma\in\Sigma$ and $C\subseteq B$, 
          satisfying one of the following conditions:
          \begin{enumerate}
            \item $q'=q''=q_\exit$,
            \item $q'\in Q_i$, $q''=q_\exit$, and the conjunctive normal form of $\delta(q',\sigma)$ contains a clause of the form $\bigvee C$,
            \item $q',q''\in Q_i$ and the conjunctive normal form of $\delta(q',\sigma)$ contains a clause of the form $(\square, q'') \lor \bigvee C$.
          \end{enumerate}
          In this case we interpret $\mathcal{A}^q_\exit$ as a universal B\"uchi automaton (\UBA),
          which accepts an $\omega$-word whenever all runs on it stay forever inside the component $Q_i$.
    %
\end{itemize}
\end{defi}

We observe that the automata $\mathcal{A}^q_\exit$ defined above are also looping.
The above definition allows to consider every non-transient component as an $\omega$-word automaton, 
while carrying the branching information through a suitable annotation of the input (i.e., the sets $C$ of atoms with states outside the component). 
This interpretation is particularly adequate for formalising restrictions directly on the linearisations of the components.
A first natural restriction is to require the linearised components to be \emph{counter-free}, in order to match the expressive 
power of \FO, at least at the level of the annotated $\omega$-word languages. We shall call an \PHTA \emph{counter-free} if the linearisation of each of its non-transient components is counter-free. 
However, this does not suffice to capture the expressive power of \FO on \emph{trees}, essentially because annotations 
may be used that conceal counting behaviour behind an apparently counter-free semantics. This will be further explained in the next example.
Towards solving this problem, we enforce an additional restriction on the use of annotations, called \emph{visibility}
(the restriction was originally introdued in \cite{BBMP24} and used in \cite{benerecetti2025automaton} under
the name of \emph{mutual exclusion}). 
Hereafter, whenever we work with an automaton $\mathcal{A}^q_\exit$ and an annotation $C$, we tacitly assume that $C$ 
appears in the transition relation $\Delta$ of $\mathcal{A}^q_\exit$, namely, there exist $q'$ in the same component 
of $q$, $\sigma\in\Sigma$, and $q''$ in a different component of $q$ such that $(q',(\sigma,C),q'') \in \Delta$.

\begin{defi}
Let $Q_i$ be a non-transient component of a \PHTA $\mathcal{A}$. 
We say that $Q_i$ is \emph{visible} if for every pair of distinct annotations $C \neq C'$ that appear
in some transitions of $\mathcal{A}^q_\exit$, for any $q\in Q_i$, 
there are atoms $\theta\in C$ and $\theta'\in C'$ for which the tree languages
$\Lang(\mathcal{A}^\theta)$ and of $\Lang(\mathcal{A}^{\theta'})$ are the complement of one another. 
Accordingly, we say that $\mathcal{A}$ is \emph{visible} if every one of its non-transient components 
is visible. 
\end{defi}

Note that if $C,C'$ are distinct annotations of a component of a visible \PHTA $\mathcal{A}$, 
then, in particular, we have that $\Lang(\mathcal{A}^{\bigwedge C}) \cap \Lang(\mathcal{A}^{\bigwedge C'})$ is empty
and, symmetrically, $\Lang(\mathcal{A}^{\bigvee C}) \cup \Lang(\mathcal{A}^{\bigvee C'})$ contains all possible trees over $\Sigma$.
These two properties will turn out useful when the component is existential, resp., universal.

We further discuss the visibility restriction with an example, adapted from \cite[Example 5.1]{BBMP24}.

\begin{exa}
Let $a,b$ be atomic propositions and $\Sigma=2^{\{a,b\}}$. 
Consider the language $L$ of trees that have a finite access path of even length such that 
\begin{enumerate}
\item the path is uniformly labeled by $\emptyset$,
\item along the path every node at an even position has a child outside the path that is labeled by $\{a\}$, and 
\item the last node of the path has a child labeled by $\{b\}$. 
\end{enumerate}
It is easy to see that \FO cannot define such a language, since capturing the above properties requires counting modulo two,
which \FO cannot do in the standard signature for unordered trees --- the situation would be different if the signature 
contained the $i$-th successor predicate (see, e.g., \cite[Example 1]{Pot95} and \cite[Section 1]{Boj21}). 
    
Now, let us define the \PHTA $\mathcal{A} = \langle \{q_I, q, q_a, q_b\}, \Sigma, \delta, q_I, \emptyset\rangle$, where 
\[
\arraycolsep=2pt
\begin{array}{rclrcl}
\delta(q_I,\sigma) &=&
\begin{cases}
    \big((\diamond, q) \land (\diamond, q_a)\big) \lor (\diamond, q_b)   & \text{if $\sigma=\emptyset$}, \\
    \bot & \text{otherwise;}
\end{cases}
\qquad\qquad
&\delta(q,\sigma) &=&
\begin{cases}
    (\diamond, q_I) & \text{if $\sigma=\emptyset$}, \\
    \bot & \text{otherwise;}
\end{cases}
\\[4ex]
\delta(q_a,\sigma) &=&
\begin{cases}
    \top & \text{if $a\in\sigma$}, \\
    \bot & \text{otherwise;}
\end{cases}
\qquad\qquad
&\delta(q_b,\sigma) &=&
\begin{cases}
    \top & \text{if $b\in\sigma$}, \\
    \bot & \text{otherwise.}
\end{cases}
\end{array}
\]
%
It is easy to see that $\mathcal{A}$ recognises precisely the language $L$.
Moreover, the state set of $\mathcal{A}$ is partitioned into three components: 
an existential component $Q_1 = \{q_I, q\}$ and two transient components $Q_2 = \{q_b\}$ and $Q_3 = \{q_a\}$. 
We can therefore linearise $Q_1$, which results in the following \NCA $\mathcal{A}^{q_I}_\exit$:
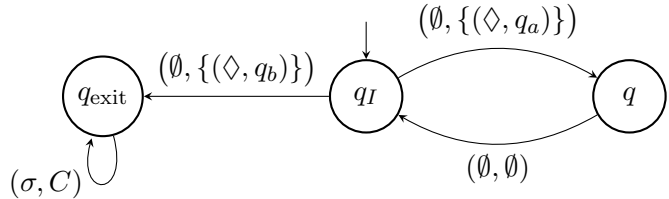
\begin{center}
\begin{tikzpicture}[
    ->,                 
    >=stealth,          
    node distance=3.5cm,  
    every state/.style={thick}, 
    initial text=$ $,   
]

    \node[state, initial above] (qI) {$q_I$};
    \node[state] (q) [right of=qI] {$q$};
    
    \node[state] (qexit) [left of=qI] {$q_\exit$}; 

    \path
        (qI) edge[bend left] node[above] {$\big(\emptyset, \{(\diamond, q_a)\}\big)$} (q)

        (q) edge[bend left] node[below] {$(\emptyset, \emptyset)$} (qI)

        (qI) edge node[above] {$\big(\emptyset, \{(\diamond, q_b)\}\big)$} (qexit)

        (qexit) edge[loop below] node[left] {$(\sigma, C)~$} (qexit);

\end{tikzpicture}
\end{center}
where $(\sigma, C)$ denotes any symbol from $\Sigma=2^{\{a,b\}}$ paired with a possible annotation 
that consists of atoms with states outside $Q_1$.
Notice that $\mathcal{A}^{q_I}_\exit$ is counter-free; this provides a concrete example where only
enforcing counter-freeness on the linearisations of the components of a \PHTA does not guarantee
that the recognised tree language is \FO definable. 
It is also worth noting that 
$\mathcal{A}$ is not visible, since $Q_1$ itself is not visible. Indeed, for $C=\emptyset$ and 
$C' = \{(\diamond, q_b)\}$, we have $\bigwedge C = \top$ (by the usual convention that an empty conjunction is equivalent to $\top$)
and $\bigwedge C' = (\diamond, q_b)$, and hence $\Lang(\mathcal{A}^\top) \cap \Lang(\mathcal{A}^{(\diamond, q_b)}) \neq \emptyset$, 
which violates the definition of visibility. Before proceeding, we clarify the problem with this automaton and explain how it arises from its annotations. The transitions from $q_I$ to $q$, $q$ to $q_I$, and from $q_I$ to $q_{\mathrm{exit}}$ carry the annotations $\{(\diamond,q_a)\}, \emptyset$, and $\{(\diamond,q_b)\}$, respectively. These require, respectively, a child satisfying $a$, no additional branching obligation, and a child satisfying $b$. The difficulty is that these annotations do not describe mutually exclusive situations. For example, a subtree with a child satisfying $a$ satisfies the requirements of both $\{(\diamond,q_a)\}$ and $\emptyset$. Consequently, along a path where every node has such a child, the first two annotations can alternate, supporting the alternation between $q_I$ and $q$. The linearised automaton treats them as distinct input symbols and is counter-free, yet the tree automaton can still enforce that the exit occurs at an even depth. This is the counting behaviour that annotations can conceal and that the visibility restriction is designed to expose.

On the other hand, if we insist in having a \emph{visible} \PHTA that recognises $L$, we would get only automata 
for which some linearisations of components are not counter-free. Indeed, we will see that it is the conjunction 
of those two properties, counter-freeness and visibility, that characterises \FO.
For example, consider the \PHTA $\mathcal{A}' = \langle Q', \Sigma, \delta', q_I, \emptyset\rangle$, where 
$Q' = \{q_I, q, q_a, q_b, q_{\neg a}, q_{\neg b}\}$ and where the images of $\delta'$ are defined as follows:
\begin{itemize}
\item $\delta' (q_I, \emptyset) = \bigvee 
  \begin{cases}
    ((\diamond, q) \land (\diamond, q_a) \land (\diamond, q_b)) \\
    ((\diamond, q) \land (\diamond, q_a) \land (\square, q_{\neg b})) \\
    ((\diamond, q_a) \land (\diamond, q_b))  \\
    ((\square, q_{\neg a}) \land (\diamond, q_b))
  \end{cases}$
\item $\delta'(q_I, \sigma) = \bot$, for every $\sigma \neq \emptyset$,
\item $\delta' (q, \emptyset) = \bigvee
  \begin{cases}
    ((\diamond, q_I) \land (\diamond, q_a) \land (\diamond, q_b))  \\
    ((\diamond, q_I) \land (\diamond, q_a) \land (\square, q_{\neg b})) \\
    ((\diamond, q_I) \land (\square, q_{\neg a}) \land (\diamond, q_b))\\
    ((\diamond, q_I) \land (\square, q_{\neg a}) \land (\square, q_{\neg b}))
  \end{cases}$
\item $\delta'(q, \sigma) = \bot$, for every $\sigma \neq \emptyset$,
\item for every $\alpha\in\{a,b,\neg a,\neg b\}$, 
      $\delta'(q_\alpha,\sigma)$ is either $\top$ or $\bot$ depending on whether $\sigma$ satisfies $\alpha$.
\end{itemize}
Despite the tedious definitions above, the intuition for this new automaton $\mathcal{A}'$ is quite simple, especially when
compared with the previous one $\mathcal{A}$. 
For example, the previous transition function $\delta$ was mapping 
$(q,\emptyset)$ to $(\diamond,q_I)$. The latter image can be seen as a formula that commits to checking certain local 
properties of the tree with states inside the same component of $q$ (i.e., $(\diamond, q_I)$) and other properties
with states outside the component of $q$ (i.e., $\top$). The new transition function $\delta'$ performs a similar
commitment, but now with refined atomic properties that are ``visible'' in the target states that exit the component of $q$.
As a matter of fact, the process of making those atomic commitments visible in a \HTA can be made automatic, and indeed 
in \cite[Proposition 5.7]{BBMP24} it was shown that every \HTA can be effectively transformed to 
an equivalent \HTA that is also visible. 

Coming back to our example, we claim that $\mathcal{A}'$ is visible. 
There are now five components: one existential component $Q'_1 = \{q_I,q\}$
and four transient components, one for each of the states $q_a,q_b,q_{\neg a},q_{\neg b}$.
The linearisation of $Q'_1$ is basically a refinement of the linearisation of the previous component $Q_1$ of $\mathcal{A}$; 
we depict this below, without drawing the uninteresting arrows towards $q_\exit$:

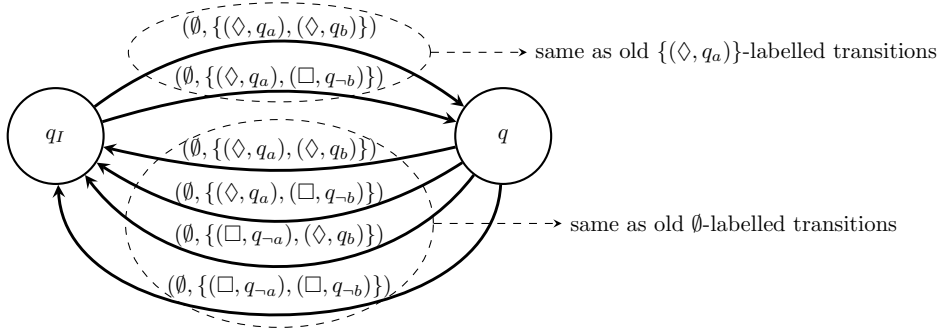
\begin{center}
\begin{tikzpicture}[
    ->,
    >=stealth,
    node distance=7cm,
    every state/.style={thick, circle, minimum size=1.5cm},
    initial text=$ $,
    auto,
    font=\small,
    scale=0.85, transform shape 
]

    \node[state] (qI) {$q_I$};
    \node[state] (q) [right of=qI] {$q$};

    \path
        (qI) edge[bend left=37, line width=1.1]
            node[above=-1mm,name=t1] {$(\emptyset,\{(\diamond,q_a),(\diamond,q_b)\})$} (q)
        (qI) edge[bend left=17, line width=1.1]
            node[above=-1mm,name=t2] {$(\emptyset,\{(\diamond,q_a),(\square,q_{\neg b})\})$} (q);

    \path
        (q) edge[bend left=13, line width=1.1]
            node[above,name=b1] {$(\emptyset,\{(\diamond,q_a),(\diamond,q_b)\})$} (qI)
        (q) edge[bend left=33, line width=1.1]
            node[above=1mm,name=b2] {$(\emptyset,\{(\diamond,q_a),(\square,q_{\neg b})\})$} (qI)
        (q) edge[bend left=53, line width=1.1]
            node[above=1.5mm,name=b3] {$(\emptyset,\{(\square,q_{\neg a}),(\diamond,q_b)\})$} (qI)
        (q) edge[bend left=87, line width=1.1]
            node[above=1mm,name=b4] {$(\emptyset,\{(\square,q_{\neg a}),(\square,q_{\neg b})\})$} (qI);

    \node[
        draw,dashed,ellipse,
        fit=(t1)(t2),
        inner xsep=-3pt,
        inner ysep=-4pt,
    ] (caseone) {};

    \node[right=15mm of caseone] (caseonetext) {same as old $\{(\diamond,q_a)\}$-labelled transitions};
    \draw[dashed] (caseone.east) -- (caseonetext.west);

    \node[
        draw,dashed,ellipse,
        fit=(b1)(b2)(b3)(b4),
        inner xsep=-5pt,
        inner ysep=-6pt,
        yshift=-2pt,
    ] (casetwo) {};

    \node[right=20mm of casetwo] (casetwotext) {same as old $\emptyset$-labelled transitions};
    \draw[dashed] (casetwo.east) -- (casetwotext.west);

\end{tikzpicture}
\end{center}
This automaton is not counter-free, because there is a path from $q_I$ to $q_I$ reading two occurrences of 
$\big(\emptyset, \{(\diamond, q_a), (\diamond, q_a)\}\big)$, but there is no path from $q_I$ to $q_I$ 
reading just one such occurrence. 
\end{exa}

In the following, we will prove that if a \PHTA is both counter-free and visible, it is expressively 
equivalent to \FO. We anticipate here that the property of being visible is also connected to the 
regular expressions for regular tree languages introduced in \cite{Tho87}, that 
will be further investigated in future work. 


\begin{thm}
Given a counter-free, visible \PHTA $\mathcal{A} = \langle Q, \Sigma, \delta, q_I, F \rangle$ over $\Sigma = 2^\mathrm{AP}$, there is a \PolCTLs formula $\varphi_\mathcal{A}$ such that $\Lang(\mathcal{A}) = \Lang(\varphi_\mathcal{A})$.
\end{thm}

\begin{proof}
Fix a counter-free, visible \PHTA $\mathcal{A} = \langle Q, \Sigma, \delta, q_I, F \rangle$. 
Similarly to the proof of Theorem \ref{PHTA-POLPCTL equivalence}, for each $\theta \in \mathcal{B}^+(\{\diamond_k, \square_k\} \times Q)$
(resp., $q\in Q$), we aim at constructing a corresponding \PolCTLs formula $\varphi^\theta$ (resp., $\varphi^q$)
such that $\Lang(\mathcal{A}^\theta) = \Lang(\varphi^\theta)$ (resp., $\Lang(\mathcal{A}^q) = \Lang(\varphi^q)$).
The desired formula that defines the original language $\Lang(\mathcal{A})$ would then be $\varphi_{\mathcal{A}} = \varphi^{q_I}$.
We will define the formulae $\varphi^\theta$ and $\varphi^q$ by a mutual induction. The case for $\theta$ is straightforward:
\begin{itemize}
\item $\varphi^\top = \top$,
\item $\varphi^\bot = \bot$,
\item $\varphi^{\theta_1 \land \theta_2} = \varphi^{\theta_1} \land \varphi^{\theta_2}$,
\item $\varphi^{\theta_1 \lor \theta_2} = \varphi^{\theta_1} \lor \varphi^{\theta_2}$,
\item $\varphi^{(\diamond_k,q)} = \D^k \varphi^q$,
\item $\varphi^{(\square_k,q)} = \neg \D^k \neg\varphi^q$.
\end{itemize}
In particular, as degenerate cases of the last two items, we get $\varphi^{(\diamond,q)} = \E\X \varphi^q$ and $\varphi^{(\square,q)} = \A\X \varphi^q$.

The rest of the proof is devoted to constructing the formulae $\varphi^q$. 
For this, we shall exploit a natural induction based on the order of the components of $\mathcal{A}$, starting from 
the smallest components (recall that the transitions of $\mathcal{A}$ can only move from one component to the same 
component or to a smaller one).
We distinguish some cases based on whether the considered state $q$ is in a transient, existential, or universal component:
\begin{itemize}
\item If $q$ belongs to a transient component $Q_i$, then by construction, for every $\sigma \in \Sigma$, 
      $\delta(q, \sigma) = \theta_{q, \sigma}$ is a positive Boolean combination of atoms that contain
      only states from components of order lower than $Q_i$. By the inductive hypothesis, we 
      know how to construct corresponding formulae $\varphi^{\theta_{q,\sigma}}$ using the previous definitions, 
      and we can then define 
      \[
        \varphi^q ~=~ \bigvee\nolimits_{\sigma \in \Sigma} \Big( \bigwedge\nolimits_{a\in\sigma} a ~\land~ \bigwedge\nolimits_{a\in\mathrm{AP}\setminus\sigma}\neg a ~\land~ \varphi^{\theta_{q, \sigma}} \Big).
      \]
    
\item If $q$ belongs to an existential component $Q_i$, then we consider the linearisation $\mathcal{A}^q_\exit$ of $Q_i$ 
      as defined in Definition \ref{Linearization}. 
      By construction, this is an \NCA $\mathcal{A}^q_\exit = \langle Q_i \uplus \{q_{exit}\}, \Sigma', \delta', q, Q_i \rangle$
      over an alphabet $\Sigma' = \Sigma \times 2^B$ that is obtained by annotating elements of $\Sigma$ with sets that contain 
      only atoms with states from components strictly smaller than $Q_i$.

      We remark a few important properties that will be used later:
      \begin{enumerate}[{P}1)]
        \item $\mathcal{A}^q_\exit$ is both counter-free and looping.
        \item For every annotation $C$ that appears in a transition of $\mathcal{A}^q_\exit$,
              $C$ contains only atoms with states from components of order lower than $Q_i$;
              we can thus activate the inductive hypothesis and obtain a formula 
              $\varphi^{\bigwedge C}$ that defines the language $\Lang(\mathcal{A}^{\bigwedge C})$. 
        \item Because $Q_i$ is visible, for every pair of distinct annotations $C,C'$
              used by $\mathcal{A}^q_\exit$, we have $\Lang(\varphi^{\bigwedge C}) \cap \Lang(\varphi^{\bigwedge C'}) = \emptyset$.  
      \end{enumerate}

      \medskip
      We begin by discussing Property P1. In view of this property, it is tempting to 
      claim a use of Lemma \ref{looping equivalence} for translating $\mathcal{A}^q_\exit$ 
      into an equivalent \cosafeLTL formula. However, this is not immediately possible because the
      alphabet of $\mathcal{A}^q_\exit$ 
      contains \emph{pairs} of the
      form $(\sigma,C)$, with $\sigma\subseteq\mathrm{AP}$ and $C\subseteq B$. It is however straightforwad 
      to encode
      every such pair $(\sigma,C)$ by the disjoint union $\sigma \uplus C$, and accordingly see the 
      alphabet $\Sigma'$ of $\mathcal{A}^q_\exit$ as a powerset of a new set $\mathrm{AP}_\exit = \mathrm{AP} \uplus B$
      of atomic propositions. To highlight this, we call this alphabet $\Sigma'_\exit$.
      Now, we can apply Lemma \ref{looping equivalence} correctly and 
      obtain a \cosafeLTL formula equivalent to $\mathcal{A}^q_\exit$, as summarised in the following claim:

      \begin{claim}
      There is a \cosafeLTL formula $\varphi^q_\exit$ that, up to the encoding of $\Sigma'$ into $2^{\mathrm{AP}_\exit}$,
      is equivalent to the \NCA $\mathcal{A}^q_\exit$.
      \end{claim}

      We have basically done the first main step towards translating $\mathcal{A}^q$ into an equivalent \PolCTLs formula $\varphi^q$. 
      To avoid any misunderstanding, we remark that the encoding that we have used above is a minor technical 
      detail, which is not directly related to the visibility assumption 
      (this assumption will be used later, to show the correctness of a certain construction based on guessing annotations ``on the fly'').

      Note that the formula $\varphi^q_\exit$ obtained above is in negation normal form and only uses the modal operators
      $\X$ and $\U$, as required by the grammar of \cosafeLTL (see Section \ref{preliminaries}). 
      We can take advantage of this normal form, to perform a further rewriting on $\varphi^q_\exit$, 
      so as to obtain a formula that, whenever uses any atomic predicate from $\Sigma'_\exit \cap B$, 
      it does so within a subformula of the form 
      \[
        \varphi^C_\exit ~=~ \bigwedge\nolimits_{\theta\in C} \theta ~\land~ \bigwedge\nolimits_{\theta\in B\setminus C} \neg\theta.
      \]
      Having enforced this stronger normal form will allow us to manipulate entire subformulae $\varphi^C_\exit$
      as if they were atoms, e.g., enabling substitutions of $\varphi^C_\exit$ by other formulae.
      
      Below, we give another important claim (without proof, since it is easy to verify), which characterises
      acceptance of a tree by $\mathcal{A}^q$ by the existence of an annotated path accepted by $\mathcal{A}^q_\exit$:

      \begin{claim}
      A tree $t$ is accepted by $\mathcal{A}^q$ if and only if there are an infinite path $\pi = v_0 v_1 \dots$ 
      from the root, spelling out the sequence of labels $\sigma_0 \sigma_1 \dots$, and a sequence of annotations $C_0,C_1,\dots$ 
      such that (1) the $\omega$-word $(\sigma_0,C_0) \: (\sigma_1,C_2) \: \dots$ is accepted by $\mathcal{A}^q_\exit$,
      and (2) for every $i\ge 0$, the subtree $t_{v_i}$ of $t$ rooted at $v_i$ is accepted by $\mathcal{A}^{\bigwedge C_i}$
      (note the latter condition for $i=0$ is weaker than asking that $t$ is accepted by $\mathcal{A}^q$).
      \end{claim}

      We recall that, in order for $\mathcal{A}^q_\exit$ to accept an annotated path like above, it must eventually 
      exit the component $Q_i$ and reach state $q_\exit$, after which $\mathcal{A}^q_\exit$ accepts no matter what.
      So the claim above highlights the fact that acceptance of an infinite tree starting from an existential
      component of a \PHTA is witnessed by the existence of a suitable finite annotated path.
      
      The next step is rather straightforward and consists of replacing the conditions in the previous claim,
      which use automata, by equivalent conditions that use formulae. Specifically, we replace condition (1) 
      by the requirement that the $\omega$-word $(\sigma_0,C_0) \: (\sigma_1,C_2) \: \dots$ satisfies the
      \cosafeLTL formula $\varphi^q_\exit$. Equivalence of this new condition is implied by Claim 1.
      Similarly, thanks to Property P2, we can replace condition (2) by the requirement that, for every 
      $i\ge 0$, the subtree $t_{v_i}$ rooted at $v_i$ satisfies $\varphi^{\bigwedge C}$.
      Overall, this allows us to rewrite Claim 2 as follows:

      \begin{claim}
      A tree $t$ is accepted by $\mathcal{A}^q$ if and only if there are an infinite path $\pi = v_0 v_1 \dots$ 
      from the root, spelling out the sequence of labels $\sigma_0 \sigma_1 \dots$, and a sequence of annotations $C_0,C_1,\dots$ 
      such that (1) the $\omega$-word $(\sigma_0,C_0) \: (\sigma_1,C_2) \: \dots$ satisfies $\varphi^q_\exit$,
      and (2) for every $i\ge 0$, the subtree $t_{v_i}$ of $t$ rooted at $v_i$ satisfies $\varphi^{\bigwedge C_i}$.
      \end{claim}

      The final step is to derive a characterisation that does not require guessing beforehand the correct 
      sequence of annotations $C_0,C_1,\dots$ along a path, since modal logics like \PolCTLs do not have this ability. 
      The idea is instead to guess ``on the fly'', while traversing the path $\pi=v_0 v_1 \dots$, each annotation $C_i$
      and check that the current subtree $t_{v_i}$ satisfies the corresponding formula $\varphi^{\bigwedge C_i}$.
      However, because subtrees are nested one inside another, this raises the following crux: 
      is it possible that guessing annotations independently at the different nodes along the path $\pi=v_0 v_1 \dots$ 
      results in an inconsistent sequence $C_0,C_1,\dots$ that cannot be realised on the given tree $t$? 
      Note that the latter issue can only arise if there are multiple sequences of annotations
      that satisfy the required formulae. 
      Luckily, this is precisely the part where we can exploit the visibility assumption, and in particular Property P3: 
      it is not possible that two distinct sequences $C_0,C_1,\dots$ and $C'_0,C'_1,\dots$
      can be used to annotate the same sequence of labels $\sigma_0 \sigma_1 \dots$ along the same path 
      $\pi = v_0 v_1 \dots$ of $t$, and yet satisfy (1) $(\sigma_0,C_0) \: (\sigma_1,C_2) \: \dots \models \varphi^q_\exit$
      and (2) $t_{v_i} \in \Lang(\mathcal{A}^{\bigwedge C_i})$ for all $i\ge 0$. Indeed, if this happened,
      then there would exist a subtree $t_{v_i}$ that witnesses a non-empty intersection between the languages
      $\Lang(\varphi^{\bigwedge C_i})$ and $\Lang(\varphi^{\bigwedge C'_i})$. 
      As a result of this, we can simplify our characterisation of acceptance and replace 
      the quantification on sequences of annotations by simpler disjunctions evaluated at nodes:

      \begin{claim}
      A tree $t$ is accepted by $\mathcal{A}^q$ if and only if there is an infinite path $\pi = v_0 v_1 \dots$ 
      from the root 
      such that
      \[
          t,\pi, 0 ~\models~ \varphi_\exit^q \big[ \varphi^C_\exit / \varphi^{\bigwedge C} \big]
      \]
      where the formula on the right hand-side is obtained from $\varphi_\exit^q$ by replacing every occurrence
      of a subformula $\varphi^C_\exit$ (recall that these are plain propositional formulae fully describing 
      a local annotation $C$) by $\varphi^{\bigwedge C}$. 
      \end{claim}

      Note that the formula written above is a path formula of \PolCTLs over the original set $\mathrm{AP}$ 
      of atomic propositions. What remains to be done is to simply add the existential path quantifier $\E$ 
      on top of this, whose semantics reproduces precisely the characterisation given above,
      and obtain in this way the desired state formula of \PolCTLs equivalent to $\mathcal{A}^q$:

      \begin{claim}
      A tree $t$ is accepted by $\mathcal{A}^q$ if and only if it satisfies the \PolCTLs state formula
      \[
          \varphi^q ~=~ \E \Big( \varphi_\exit^q \big[ \varphi^C_\exit / \varphi^{\bigwedge C} \big] \Big).
      \]
      \end{claim}

\item It remains to consider the case where $q$ is a state from a universal component $Q_i$.
      Here we apply a standard dualisation argument to reduce this case to the existential one. 
      Let $\overline{\mathcal{A}} = \langle Q, \Sigma, \overline{\delta}, q_I, \overline{F}\rangle$
      be the \emph{dual} automaton of $\mathcal{A}$, defined by letting
      $\overline{F} = Q \setminus F$, and for every $q \in Q$ and $\sigma \in \Sigma$, 
      $\overline{\delta}(q, \sigma) = \overline{\delta(q, \sigma)}$, where $\overline{\theta}$ 
      is the dual of the positive Boolean formula $\theta$, 
      defined inductively as follows:
\[
\begin{array}{ccc}
    \overline{\top} = \bot & \qquad\qquad
    \overline{\theta_1 \land \theta_2} = \overline{\theta_1} \lor \overline{\theta_2} & \qquad\qquad
    \overline{(\diamond_k,q)} = (\square_k, q) \\[2ex]
    \overline{\bot} = \top & \qquad\qquad
    \overline{\theta_1 \lor \theta_2} = \overline{\theta_1} \land \overline{\theta_2} & \qquad\qquad
    \overline{(\square_k,q)}  = (\diamond_k, q).
\end{array}
\]
     $\mathcal{A}$ and $\overline{\mathcal{A}}$ clearly share the same structural constraints, and every existential (resp., universal) component in $\mathcal{A}$ is a universal (resp., existential) component in $\overline{\mathcal{A}}$. Moreover, if $\mathcal{A}$ is counter-free and visible, $\overline{\mathcal{A}}$ is too, since no modification is performed at the structural level of the components. 
     Since the component of $q$ in $\mathcal{A}$ is universal, the corresponding component in $\overline{\mathcal{A}}$ is existential.
     We can thus apply our construction for the existential case. 
     Moreover, by an easy adaptation of \cite[Lemma 5.2]{KVW00}, we know that
     $\overline{\mathcal{A}^q}$ recognises the complement of the language $\Lang(\mathcal{A}^q)$, and
     hence the negation of $\varphi_{\overline{\mathcal{A}^q}}$ is the desired formula. More generally, \cite[Lemma 5.2]{KVW00} establishes a complementation result for hesitant automata whose acceptance condition combines B\"uchi and coB\"uchi requirements. The construction is analogous to the one above, differing only in the treatment of the acceptance condition, where the roles of the B\"uchi and coB\"uchi requirements are exchanged. Indeed, that argument even simplifies in our setting, since the automata are weak: every infinite branch eventually remains within a component whose states are either all accepting or all rejecting. Replacing $F$ with $Q\setminus F$ therefore reverses the acceptance outcome of every infinite branch. Together with the dualisation of the transition function, which exchanges existential and universal choices, this yields an automaton recognising the complement of the original language.
     \qedhere
    \end{itemize}
\end{proof}

We conclude the section by providing the translation in the other direction,
from \PolCTLs to a counter-free visible \PHTA.

\begin{thm}
Given a \PolPCTL formula $\varphi$ over a set of atomic propositions $\mathrm{AP}$, there is a counter-free visible \PHTA $\mathcal{A}_\varphi$ such that $\Lang(\varphi) = \Lang(\mathcal{A}_\varphi$). 
\end{thm}


\begin{proof}
We follow \cite[Sec.~5.1]{BBMP24} and employ a \emph{simple form} of \PolCTLs, in which state formulae of the form $\E\psi$ in the grammar of \PolCTLs are replaced by equivalent state formulae of the form $\D^1\E\psi$. This simple form is expressively complete for \PolCTLs, as it follows by a straightforward adaptation of \cite[Proposition B.4]{BBMP24}, and it is needed to enforce the visibility property. 
The proof goes by induction on the structure of \PolCTLs state formulae in simple form. 
In particular, we consider a \PolCTLs state formula $\psi$ in simple form, and construct 
a corresponding \PHTA $\mathcal{A}_\psi$, by distinguishing the following cases:
\begin{itemize}
    \item $\psi = p$. An automaton equivalent to an atomic proposition $p$ just needs one state. We can set $\mathcal{A}_\psi = \langle \{q_I\}, \Sigma, \delta, q_I, \emptyset\rangle$, with $\delta(q_I, \sigma) = \top$ or $\delta(q_I, \sigma) = \bot$ depending on whether $p \in \sigma$ or not. 
    \item $\psi = \neg \psi'$. By inductive hypothesis, there is an automaton $\mathcal{A}_{\psi'}$ equivalent to $\psi'$. Then, it suffices to construct the dual automaton $\overline{\mathcal{A}_{\psi'}}$, which recognises the complement language of $\Lang(\mathcal{A}_{\psi'})$. 
    \item $\psi = \psi_1 \lor \psi_2$. By inductive hypothesis, we can construct two automata $\mathcal{A}_{\psi_1} = \langle Q_1, \Sigma, \delta_1, q_1, F_1 \rangle$ and $\mathcal{A}_{\psi_2} = \langle Q_2, \Sigma, \delta_2, q_2, F_2 \rangle$, equivalent to $\psi_1$ and $\psi_2$, respectively. 
        We then define $\mathcal{A}_\psi = \langle Q_1 \uplus Q_2 \uplus \{q_I\}, \Sigma, \delta, q_I, F_1 \uplus F_2 \rangle$,
        where $q_I$ is a fresh state. 
        The new state set $Q_1 \uplus Q_2 \uplus \{q_I\}$ inherits the partitions into components from $\mathcal{A}_{\psi_1}$ and $\mathcal{A}_{\psi_2}$, plus a singleton, transient component $\{q_I\}$. The types of components for $Q_1$ and $Q_2$ are the same as in the original automata. The ordering of the components is defined hierarchically: the singleton component $\{q_I\}$ is the highest-order component, then come the components of $Q_1$, followed by those of $Q_2$. 
        Note that the set of final states of $\mathcal{A}_{\psi}$ is the union of the final states of $\mathcal{A}_{\psi_1}$ and $\mathcal{A}_{\psi_2}$.
        Finally, the transition function $\delta_\psi$ coincides with $\delta_1$ (resp., $\delta_2$) when applied to states in $Q_1$ (resp., $Q_2$), and maps $(q_I,\sigma)$, for any $\sigma \in \Sigma$, to the positive Boolean formula 
        $\delta_\psi(q_1, \sigma) \lor \delta_\psi(q_2, \sigma)$. 
                
        It is easy to see that the automaton $\mathcal{A}_\psi$ recognises precisely the union of the languages of $\mathcal{A}_{\psi_1}$ and $\mathcal{A}_{\psi_2}$, while still remaining counter-free and visible.
    \item $\psi = \D^n \psi'$. By inductive hypothesis, we can construct an automaton $\mathcal{A}_{\psi'} = \langle Q', \Sigma, \delta', q'_I, F' \rangle$ equivalent to $\psi'$. 
          We define $\mathcal{A}_\psi=\langle Q \uplus \{q_I\}, \Sigma, \delta, q_I, F' \rangle$ as follows.
          The state set is obtained by extending the original set $Q$ with a fresh state $q_I$, which is the new initial state of $\mathcal{A}_\psi$.
          This state set inherits the partition into components from the original automaton $\mathcal{A}_{\psi'}$, and a new singleton, transient component $\{q_I\}$ is added.
          The same set $F'$ of final states is adopted as set of final states in $\mathcal{A}_{\psi}$.
          Moreover, the transition function $\delta$ is the extension of $\delta'$ that maps any pair $(q_I,\sigma)$, with $\sigma\in\Sigma$, to $(\diamond_n, q'_I)$. 

          Intuitively, the automaton $\mathcal{A}_\psi$ starts at the root and immediately selects $n$ children along which it sends copies of $\mathcal{A}_{\psi'}$ that check satisfaction of the subformula $\psi'$.
          It is then straightforward to check that $\mathcal{A}_\psi$ recognises precisely the language defined by $\psi = \D^n\psi'$.

    \item $\psi = \E \gamma$. While in the previous cases we could use the inductive hypothesis, in the case of $\psi = \E\gamma$ we cannot, 
          since $\gamma$ is not a \emph{state} subformula, but a \emph{path} subformula. 
          We proceed by following the classical approach of substituting problematic state subformulae by new atomic propositions. The construction below follows that of \cite[Theorem 5.3]{KVW00}, with adaptations accounting for the strictly weaker expressive power of our components compared with those of the hesitant automata considered there.
          First, we identify \emph{maximal} state subformulae of $\gamma$ of the form $\E\gamma_1$ or $\D^n \psi_1$, and collect them in the set $\max(\gamma)$.
          Formally, $\max(\gamma)$ contains all and only those formulae $\E\gamma_1$ or $\D^n \psi_1$ that occur as subformulae of $\gamma$ but are not themselves
          contained in larger subformulae of $\varphi$ of the form $\E\gamma_1$ or $\D^n \psi_1$.

          By the assumption on $\varphi$ being in simple form, we deduce that no formula in $\max(\gamma)$ is of the form $\E\gamma_1$. Hence, we have $\max(\gamma) = \{\D^{n_1}\psi_1, ..., \D^{n_k}\psi_k\}$. 
          By the inductive hypothesis, for every $i\le k$, we can construct an automaton $\mathcal{A}_{\psi_i} = \langle Q_i, \Sigma, \delta_i, q_i, F_i \rangle$ equivalent to $\psi_i$. 
          We then consider the dual versions $\overline{\mathcal{A}_{\psi_i}} = \langle \overline{Q_i}, \Sigma, \overline{\delta_i}, \overline{q_i}, \overline{F_i} \rangle$ of these automata. 
          Without loss of generality, we can assume that all state sets $Q_1,\dots,Q_k,\overline{Q_1},\dots,\overline{Q_k}$ are pairwise disjoint.

          We continue by introducing a fresh atomic proposition $p_i$ for each formula $\D^{n_i} \psi_i$ in the set $\max(\gamma)$. We let $\mathrm{AP}' = \mathrm{AP} \uplus \{p_1,\dots,p_k\}$,
          and we replace in $\gamma$ every occurrence of $\D^{n_i} \psi_i$ (for $i=1,\dots,k$) with the corresponding atomic proposition $p_i$, obtaining a formula $\gamma'$ over the extended
          signature $\mathrm{AP}'$. By construction, thanks to the form of the grammar of \PolCTLs, the formula $\gamma'$ turns out to be a \cosafeLTL formula. 
          By Lemma \ref{lem:UBAcosafe}, there is a counter-free, looping \NCA $\mathcal{A}_{\gamma'}$ equivalent to $\gamma'$.  We now claim the following (the proof is simple and thus omitted):

          \begin{claim}
          A tree $t$ satisfies $\E\gamma$ if and only if there are an infinite path $\pi$ in $t$ and an $\omega$-word 
          $w \in \Lang(\mathcal{A}_{\gamma'})$ (note that this word is over the extended alphabet $\Sigma' = 2^{\mathrm{AP}'}$) 
          such that, for every $j \geq 0$:
          \begin{enumerate}
          \item $w(j) \cap \mathrm{AP} = t(\pi(j))$, namely, the restriction of the symbol $w(j)$ to the original set $\mathrm{AP}$ 
                of atomic propositions coincides with the label at position $j$ in $\pi$ according to $t$,
          \item for every $i=1,\dots,k$, $p_i \in w(j)$ iff $t,\pi(j) \models \D^{n_i} \psi_i$.
          \end{enumerate}
          \end{claim}

        Intuitively, in the above claim the $\omega$-word $w$ encodes entirely the information available along the path $\pi$ that is relevant for checking satisfaction of
        the formula $\E \gamma$. 
        Accordingly, the desired automaton for $\E\gamma$ will guess a path $\pi$ and simulate the \NCA $\mathcal{A}_{\gamma'}$ along this path, while sending copies of 
        the automata $\mathcal{A}_{\psi_i}$ and $\overline{\mathcal{A}_{\psi_i}}$ (for every $i \in [1, k]$) to check satisfaction of the subformulae $\D^{n_i}\psi_i$. 
        Formally, we define $\mathcal{A}_\psi = \langle Q \cup \bigcup\nolimits_{i=1,\dots,k} (Q_i \cup \overline{Q_i}), \Sigma, \delta_\psi, q_I, F \rangle$ as follows.
        The set $Q$ is the state set of the \NCA $\mathcal{A}_{\gamma'}$, and forms a single existential component inside $\mathcal{A}_\psi$.
        The other sets $Q_i$ and $\overline{Q_i}$ inherit as usual the partition into components from their respective automata $\mathcal{A}_{\psi_i}$ and $\overline{\mathcal{A}_{\psi_i}}$.
        As for the ordering on components, $Q$ is the component of highest order, and the other components are ordered as in their respective automata, with some arbitrary choices
        for components belonging to different automata. 
        The set $F$ of final states of $\mathcal{A}_\psi$ is the union of the sets of final states of the automata $\mathcal{A}_{\psi_i}$ and $\overline{\mathcal{A}_{\psi_i}}$, for every $i \in [1, k]$.
        Finally, the transition function $\delta_\psi$ is the union of the transition functions $\delta_i$ and $\overline{\delta_i}$ on their respective domains $Q_i$ and $\overline{Q_i}$
        (for $i=1,\dots,k$), extended so as to satisfy, for every $q\in Q$ and $\sigma\in\Sigma$,
        \[
        \delta_\psi(q, \sigma) ~=~ 
        \bigvee\nolimits_{I \subseteq [1, k]} 
        \bigvee\nolimits_{q' \text{ s.t.~} (q,\sigma^+,q')\in \Delta} 
        \Big((\diamond, q') ~\land~ \bigwedge\nolimits_{i \in I} (\diamond_{n_i}, q_i) ~\land~ \bigwedge\nolimits_{i \in [1,k] \setminus I} (\square_{n_i}, \overline{q_i})\Big) 
        \]
        where $\sigma^+ = \sigma \cup \bigcup_{i \in I} \{p_i\}$ and $\Delta$ is the transition relation of the \NCA $\mathcal{A}_{\gamma'}$. 
        In other words, the above transition function allows the automaton $\mathcal{A}_\psi$ to read an input symbol $\sigma$ from the current state $q$,
        guess an annotation $\sigma^+$ of $\sigma$ that represents, for each $i=1,\dots,k$, whether $\D^{n_i} \psi_i$ or its negation holds at the current subtree, 
        and accordingly simulate a $\sigma^+$-labeled transition of the \NCA $\mathcal{A}_{\gamma'}$ from $q$ to $q'$ while moving towards a child, eventually
        witnessing a path in the input tree that satisfies $\gamma$.
        
        By the above claim and thanks to the fact that each $\mathcal{A}_{\psi_i}$ (resp., $\overline{\mathcal{A}_{\psi_i}}$) is a \PHTA equivalent to $\D^{n_i} \psi_i$ (resp., $\neg\D^{n_i} \psi_i$),
        we obtain that the thus constructed automaton $\mathcal{A}_\psi$ is visible and equivalent to $\psi = \E\gamma$. 
        To prove that the automaton is also counter-free, we need to show that the linearisation of each component of $\mathcal{A}_\psi$ is counter-free. 
        We recall that the components of $\mathcal{A}_\psi$ are basically given by the components of the automata $\mathcal{A}_{\psi_i}$ and $\overline{\mathcal{A}_{\psi_i}}$,
        whose linearisations are already counter-free, 
        plus an additional existential component $Q$.
        Let $\mathcal{A}$ be the linearisation of the latter component $Q$.
        The crucial observation is that $\mathcal{A}$ is isomorphic to the \NCA $\mathcal{A}_{\gamma'}$, which we know is already counter-free.
        Indeed, the only difference between the two automata $\mathcal{A}$ and $\mathcal{A}_{\gamma'}$ is in their alphabets:
        by definition of linearisation, the former automaton is over the alphabet $\Sigma \times 2^B$, where $B$ is the set of atoms with states in $\bigcup\nolimits_i(Q_i \cup \overline{Q_i})$,
        and the latter automaton is over the alphabet $\Sigma' = 2^{\mathrm{AP}'} = 2^{\mathrm{AP} \uplus \{p_1,\dots,p_k\}} = \Sigma \times 2^{\{p_1,\dots,p_k\}}$.
        There is a natural bijection between the subsets $P\subseteq \{p_1,\dots,p_k\}$ and the subsets $C\subseteq B$, which maps any $P\subseteq \{p_1,\dots,p_k\}$
        to the set $C = \bigcup\nolimits_{i \text{ s.t.~} p_i\in P} (\diamond_{n_i},q_i) \:\cup\: \bigcup\nolimits_{i \text{ s.t.~} p_i\notin P} (\square_{n_i},\overline{q_i})$.
        Moreover, it is routine to verify that this bijection is an isomorphism between the automata $\mathcal{A}$ and $\mathcal{A}_{\gamma'}$,
        in the sense that it maps transitions of $\mathcal{A}$ to corresponding transitions of $\mathcal{A}_{\gamma'}$, while preserving the source and target states.
        This shows that the linearisation of $Q$ is counter-free, and hence $\mathcal{A}_{\psi_i}$ is a counter-free \PHTA.
        \qedhere
        \end{itemize} 
\end{proof}

The above two theorems let us conclude the following. 

\begin{thm} \label{cor:visiblectl}
    Counter-free visible \PHTA and \PolCTLs are equivalent formalisms.
\end{thm}

By the equivalence between \FO and \PolCTLs, which follows from \cite{BBMP23} and Theorem \ref{normal form CTLsf}, one can also conclude the following.

\begin{cor}
    Counter-free visible \PHTA and \FO are equivalent formalisms.
\end{cor}

Finally, going through the respective equivalence with \FO, there is a final result concerning equivalence of automata.

\begin{cor} \label{automata equivalence}
    Counter-free visible \PHTA and two-way linear \PHTA are equivalent formalisms. 
\end{cor}

Notice that the translation between counter-free visible \PHTA and two-way linear \PHTA is not direct, 
and the structure of two equivalent instances of these automata models can be very different.
In particular, it is not clear how to move efficiently from one model to another. 
We leave for the next and final section further considerations about the results here obtained.


\section{Discussion} \label{discussion}

In this paper, we investigated the expressive power of \FO over infinite trees.
We first considered two branching-time temporal logics, already known to be
equivalent to \FO, and studied their behaviour and structural properties.
We highlighted a common \emph{polarisation property}, relating it to recent work
on Weak Chain Logic.
We also provided two automaton-based characterisations of \FO over infinite
trees: the first is based on a two-way automaton whose state space can be
partitioned into singletons, while the second relies on a one-way automaton
whose state space can be partitioned into components of arbitrary size, subject,
however, to strong expressiveness restrictions when viewed as word automata.

We conclude this work by discussing some consequences of our
results and a number of related open questions.

\paragraph{Finite trees}

It is natural to investigate the implications of our results in the setting of
finite trees.
To date, no automaton-based characterisation of \FO over finite trees is known.
To be precise, in his Ph.D. thesis~\cite[Theorem 2.5.7]{Boj04},
Boja\'nczyk showed that \FO is equivalent to the \emph{cascade product of
aperiodic wordsum automata} over finite binary trees. This characterisation
implies that every deterministic automaton recognising an \FO-definable language
must admit such a decomposition.
However, no general \emph{decomposition theorem} for deterministic tree automata
is known, in contrast with the Krohn-Rhodes theorem for word
automata~\cite{krohn1965algebraic}.
As a consequence, given a deterministic tree automaton, there is currently no
effective way to determine whether it admits such a decomposition and, hence,
whether it defines an \FO language.
Let us recall that, over finite trees, deterministic bottom-up tree automata
(\ie, automata scanning the tree from the leaves to the root) recognise all
regular tree languages, unlike the infinite-tree setting, where no class of
deterministic automata captures all regular languages.
Our characterisations rely, indeed, on \emph{alternating} automata over
\emph{unranked} trees, suggesting that variants of our constructions could be
considered over finite trees with arbitrary branching.
However, several obstacles immediately arise.
Most notably, the acceptance condition becomes problematic: in the alternating
setting over finite trees, one may define acceptance in terms of runs whose
branches end in a trivial condition $\top$, making the polarised acceptance
condition vacuous.
For this and related reasons, it is unclear how to obtain an alternating
automata characterisation of \FO over finite trees, even though it seems
reasonable to believe that such a characterisation would be just an adaptation
of the ones presented here.
Nevertheless, the central challenge in the finite-tree setting is to obtain a
characterisation in terms of \emph{deterministic} automata.
Such a result would have major consequences, as it would resolve the
long-standing problem of deciding whether a regular tree language definable in
\MSO is already definable in \FO.

\paragraph{The role of visibility.} 
A natural question concerns the complexity of deciding if a given \PHTA is visible. There is a fairly straightforward EXPTIME algorithm to check if a given \PHTA is visible. Indeed, visibility requires that, for each non-transient component, every two distinct annotations contain complementary atoms, one from each annotation. Two atoms recognise complementary tree languages precisely when no tree satisfies both and no tree satisfies neither. Using intersection and dualisation, these conditions reduce to two emptiness tests for graded alternating parity tree automata of polynomial size, each decidable in EXPTIME \cite[Theorem 4.2]{BLMV08}. Since expanding the transition formulae into disjunctive (or conjunctive) normal form produces at most exponentially many annotations, checking all pairs of annotations still takes exponential time overall. Nevertheless, as already mentioned, it is also always possible to transform a given non-visible automaton into a visible one, thanks to \cite[Proposition 5.7]{BBMP24}.

\paragraph{Simulation of the two classes of automata}

In Corollary~\ref{automata equivalence}, we established the equivalence between
counter-free distinct \PHTA and two-way linear \PHTA via a logical detour.
A natural question is whether a direct translation between these two classes
exists.
This appears to be highly non-trivial: the structural restrictions imposed on
the two models are fundamentally different, and there is no obvious way to
transform, for instance, a two-way automaton with singleton components into a
one-way automaton whose components correspond to counter-free word automata.
A direct construction, although likely to be technically involved, could provide
further insight into the structural constraints underlying \FO over infinite
trees.




\paragraph{The expressive power of \PCTL}

We have already mentioned that the authors of \cite{moller2003counting} proved
\CTLs equivalent to \MPL.
In contrast, this work focuses on Weak \MPL, namely \FO, which was shown by
Schlingloff to be equivalent to \PolPCTL.
We proved that \PolPCTL is strictly less expressive than \PCTL
(Theorem~\ref{polctl expressiveness}).
Since \PCTL is itself strictly less expressive than \CTLs (and hence \MPL), it
follows that \PCTL lies strictly in between \FO and \MPL.
Identifying the fragment of \MSO (or, rather, of \WMSO) corresponding to \PCTL thus emerges as a
natural and interesting problem, as it would yield a fragment strictly
intermediate between \FO and \MPL in expressive power.




\section*{Acknowledgments}
This work was partially supported by the INdAM-GNCS project ``Neurosymbolic semantic contract for rigorous reachability and certificates in cyber-physical systems (CPS)'' (CUP E53C25002010001) and by the European Union projects ``MAPSART -- Massive Parallelism Serving Automated Reasoning Tasks'' (CUP C63C22000770006) and ``SecCo-OC -- Secure Containers Open Call'' (CUP D33C22001300002).
\bibliographystyle{alphaurl}
\bibliography{references,ReferencesFM}

\end{document}